\RequirePackage{fix-cm}
\documentclass[11pt]{article}

\usepackage[a4paper,margin=1in]{geometry}
\usepackage{amsmath,amssymb,amsthm,mathtools}
\usepackage{enumitem}
\usepackage{setspace}
\usepackage{microtype}
\usepackage{booktabs}
\usepackage{array}
\usepackage[hidelinks]{hyperref}
\usepackage[nameinlink,capitalize,noabbrev]{cleveref}
\usepackage[round,authoryear]{natbib}

\newtheorem{theorem}{Theorem}[section]
\newtheorem{proposition}[theorem]{Proposition}
\newtheorem{lemma}[theorem]{Lemma}
\newtheorem{corollary}[theorem]{Corollary}
\newtheorem{axiom}{Axiom}
\newtheorem{definition}[theorem]{Definition}
\theoremstyle{remark}
\newtheorem{remark}[theorem]{Remark}

\newcommand{\R}{\mathbb R}
\newcommand{\1}{\mathbf 1}

\title{\textbf{Utility-Level-Dependent Ambiguity}}

\author{Kemal Ozbek\thanks{Department of Economics, University of Southampton, University Road, Southampton, S017 1BJ, United Kingdom. Email: mkemalozbek@gmail.com}}
\date{\today}

\begin{document}

\maketitle

\begin{abstract}
Experimental evidence suggests that ambiguity-sensitive choice can vary systematically with the circumstances of a decision. This paper isolates one channel within a stable preference relation: ambiguity weighting may depend on the act's certainty-equivalent level. After the standard Anscombe--Aumann calibration of consequence utility, a set of behavioral axioms yields a unique continuous family of normalized monotone capacities $\{\nu_v\}_{v\in(0,1)}$. Each nonendpoint act is evaluated by the Choquet integral associated with the capacity at its own interior certainty-equivalent level, while nonendpoint acts on the same indifference surface share the same capacity. Binary event comparisons identify local event weights at each elicited level and trace their cross-level variation, providing tests of the fixed-capacity restriction. Local uncertainty aversion is equivalent to convexity of $\nu_v$ and yields an implicit multiple-priors representation with certainty-equivalent-indexed local cores. Certainty translation invariance holds if and only if the capacity is fixed across levels, recovering the maintained nondegenerate fixed-capacity Choquet expected utility benchmark; global mixture-betweenness yields implicit additive utility, and imposing both restrictions recovers full-support subjective expected utility. The capacity schedule is a reduced-form ambiguity weighting whose variation may reflect changes in ambiguity perception, ambiguity attitude, or both.
\end{abstract}

\section{Introduction}
\label{sec:introduction}

A growing experimental literature points to a broad descriptive fact:
ambiguity-sensitive choice varies with the circumstances in which uncertainty
is faced. Most directly for the present question, \citet{BaillonPlacido2019}
reject constancy of ambiguity aversion as decision makers become better off
overall and find substantial evidence in the direction of decreasing ambiguity
aversion. Related experiments document variation with monetary stakes
\citep{BouchouichaEtAl2017}, across gain--loss and likelihood domains
\citep{BaillonBleichrodt2015}, and in certainty-independence restrictions
retained by prominent ambiguity models
\citep{TrautmannWakker2018,KonigKerstingKopsTrautmann2023}. These findings need
not share a single mechanism, and not all isolate the certainty-equivalent
level from other channels. Taken together, however, they motivate asking
whether ambiguity evaluation itself can vary in a systematic and behaviorally
disciplined way across certainty-equivalent levels.

This question creates a natural contrast with benchmark theories in which the
central object governing ambiguity evaluation is globally fixed. Choquet
expected utility uses a single capacity \citep{Schmeidler1989}; maxmin expected
utility uses a fixed set of priors \citep{GilboaSchmeidler1989}; variational
preferences use a global penalty over priors
\citep{MaccheroniMarinacciRustichini2006}; and smooth ambiguity uses a stable
second-order specification together with an ambiguity-attitude aggregator
\citep{KlibanoffMarinacciMukerji2005}. Such models may still generate
context-dependent observed choice through nonlinear utility, aggregation, or
other channels. The narrower distinction is that their central
ambiguity-evaluation object is not itself indexed by the decision maker's
certainty-equivalent level. Important exceptions allow ambiguity attitudes or
ambiguity evaluation to vary with wealth, indifference classes, or endogenous
cognitive choices; these are discussed below. Relatively little is known,
however, about when observed choice reveals a systematically varying
\emph{event-weighting object}, potentially nonadditive, while keeping that
variation behaviorally disciplined and identified.

This paper studies precisely this possibility. It asks whether the ambiguity
weighting relevant to a decision can vary with the decision maker's own
certainty-equivalent level. The answer is a family
\(
    v\longmapsto\nu_v,
\)
where each \(\nu_v\) is a capacity common to all acts on the level-\(v\)
indifference surface. Thus the model does not merely permit context-dependent
ambiguity behavior: it makes the local ambiguity weighting a revealed-preference
object and allows its cross-level variation to be traced from choice. The
interpretation is deliberately reduced form. At the level of the main theorem,
$\nu_v$ is most precisely a certainty-equivalent-indexed event-weighting object
that may be nonadditive. A change in $\nu_v$ may reflect a change in perceived
ambiguity, a change in attitude toward ambiguity, or both; the baseline
representation does not claim to separate these channels. The ambiguity-averse
interpretation becomes strongest in the convex-capacity subclass characterized
by local uncertainty aversion below.

The main representation gives this question a sharp answer. As usual in the
Anscombe--Aumann framework \citep{AnscombeAumann1963}, preferences over constant acts of lotteries can be
used to recover a common affine von Neumann--Morgenstern utility scale under
the standard risk-side axioms. Since this construction is standard, we take
the resulting scale as given and focus on the uncertainty component of
preferences. After state-contingent lotteries have been expressed on this
common scale, let $x$ denote the vector of state utilities and let $V(x)$ denote
its certainty equivalent. The main result then characterizes ambiguity
evaluation on this reduced domain. Under continuity, statewise monotonicity,
strict ordering of constant acts, \emph{comonotonic mixture-betweenness}, and
a boundary nondegeneracy condition, there exists a unique capacity $\nu_v$ at
every interior certainty-equivalent level $v$ such that
\(
    x\succeq v\mathbf1
    \;\text{if and only if}\;
    C_{\nu_v}(x)\geq v.
\)
The resulting certainty-equivalent map represents preferences on the full act
domain: $V(\mathbf0)=0$, $V(\mathbf1)=1$, and every other act satisfies
\(V(x)=C_{\nu_{V(x)}}(x)\). The capacity is therefore \emph{local to the
indifference level}: two acts with the same interior certainty equivalent are
evaluated by the same event weighting, while acts at different levels may be
evaluated by different capacities. This yields a disciplined form of
level-dependent event weighting, richer than a fixed-capacity Choquet model
but substantially more structured than an arbitrary act-dependent capacity.

The axiomatic construction combines two familiar departures from subjective
expected utility. Betweenness allows the local supporting evaluation to vary
across indifference levels, following \citet{Dekel1986} and
\citet{Payro2025}, while restricting mixture-betweenness to comonotonic acts
allows local state weights to depend on ranks. The distinctive step is
cross-rank consistency at each fixed certainty-equivalent level: rank-specific
local evaluations must assign the same aggregate weight to the same event.
Binary event comparisons impose this consistency, yielding a unique capacity
$\nu_v$ and a direct revealed-preference interpretation of its event weights.
Indifference-class dependence itself has precedents, in particular,
\citet{GrantRoordaYang2025} allow ambiguity aversion to depend on the
indifference class, so the contribution is the subjective-event capacity
construction, its cross-rank consistency, and its level-by-level
identification rather than level dependence by itself.

The representation also makes it possible to study ambiguity \emph{across
certainty-equivalent levels}. First, the paper introduces local uncertainty
aversion at level $v$: mixtures of any two acts indifferent to $v\mathbf1$ are
weakly preferred to the certainty level $v$. This condition is equivalent to
concavity of the local Choquet functional and hence to convexity of $\nu_v$.
When it holds, the capacity has a nonempty core $\mathcal P_v$ and
\( C_{\nu_v}(x)=\min_{p\in\mathcal P_v}p\cdot x\). The implicit Choquet model
therefore becomes an \emph{implicit multiple-priors model}:
\( V(x)=\min_{p\in\mathcal P_{V(x)}}p\cdot x\). Unlike
\citet{GilboaSchmeidler1989}, the locally relevant prior set can expand,
contract, or otherwise change with the certainty-equivalent level.

Second, the model gives a direct revealed-preference meaning to these changes.
For an event $A$, binary event comparisons identify the cutoff gain--loss
tradeoff around level $v$. Whenever this cutoff is interior, an indifference
\(
    (v+a)\mathbf1_A+(v-b)\mathbf1_{A^c}
    \sim v\mathbf1
\)
yields
\(
    \nu_v(A)=\frac{b}{a+b}.
\)
Thus $v\mapsto\nu_v(A)$ is the decision maker's local willingness-to-bet
profile for event $A$, including at extreme event weights where the cutoff is
identified from one-sided binary comparisons rather than an interior
indifference. In the convex-capacity subclass, a pointwise fall in the
capacity as $v$ rises is equivalent to expansion of the core, providing a
precise sense in which the locally relevant set of probability models becomes
larger. Any nontrivial such uniform event ordering is intrinsically nonadditive: two
additive probability measures cannot satisfy $p_w(A)\leq p_v(A)$ for every
event unless they are identical.

Third, the paper identifies the behavioral content of a \emph{fixed} capacity.
Within the implicit-Choquet class, certainty translation invariance,
\(
    V(x+c\mathbf1)=V(x)+c,
\)
holds if and only if $\nu_v$ is independent of $v$. Within the maintained
nondegenerate class, fixed-capacity Choquet expected utility is therefore
exactly the special case in which uncertainty evaluation is invariant to
feasible common certainty translations. When $\nu_v$ varies, the same
uncertain utility increment can therefore be evaluated differently at different
certainty-equivalent levels even though the underlying vNM utility over
lotteries is unchanged. In the convex-capacity subclass this difference can
also be read through changes in the locally relevant prior set. A close
cognitive comparison is \citet{PayroTakeokaXia2026}, who endogenize ambiguity
through costly capacity choice while maintaining certainty-translation
invariance. The present model instead permits ambiguity evaluation to vary
with the certainty-equivalent level; Section \ref{sec:related} develops the
comparison.

The paper proceeds as follows. Section \ref{sec:framework} introduces the
framework and Choquet notation. Section \ref{sec:representation} states the
axioms and the representation theorem; its auxiliary results and proofs are
collected in the appendices. Section \ref{sec:special-cases} develops benchmark restrictions that locate
full-support subjective expected utility, implicit additive utility, and the
fixed-capacity Choquet model within the representation. Section
\ref{sec:ambiguity-analysis} develops the ambiguity implications, including
local uncertainty aversion, certainty-equivalent-indexed multiple priors,
revealed-preference identification, the empirical content and falsifiability of
the representation, cross-level ambiguity, translation invariance and background shifts, and
boundary qualifications. Section \ref{sec:related} then
positions the model relative to implicit and betweenness representations,
Choquet and consequence-dependent weighting, ambiguity-attitude models,
cognitive endogenous ambiguity, and the multiple-priors literature. Section
\ref{sec:conclusion} concludes. All proofs are collected in the Appendices.

\section{Framework}
\label{sec:framework}

Let $S=\{1,\ldots,n\}$, with $n\geq2$, be a finite state space. An
Anscombe--Aumann act assigns a lottery over a finite prize set to every state.
Under the standard risk-side axioms, preferences over constant acts recover a
common affine von Neumann--Morgenstern utility index over lotteries. Since that
calibration is standard, the analysis begins from the resulting utility scale.
Thus, if $f$ is an act and $u$ is this affine utility index, write
\(
    x_s(f):=u(f(s)).
\)
Normalize the range of $u$ to $I=[0,1]$. An act can then be identified with a
state-utility vector
\(x=(x_1,\ldots,x_n)\in X:=I^S\).
The constant act delivering utility $c$ in every state is denoted by
$c\mathbf1$.

This reduction isolates the uncertainty component that is the focus of the
paper: the results below characterize ambiguity evaluation once consequences
have been expressed on the standard common affine scale. If instead different
state-dependent utility indices $u_s$ are used, an interstate cardinal
normalization is needed before comparisons such as $x_s>x_t$ and
comonotonicity are invariantly defined; see the discussion of state-dependent
affine utilities in Section \ref{sec:ambiguity-analysis}.

Two vectors $x,y\in X$ are \emph{comonotonic} if \((x_s-x_t)(y_s-y_t)\geq0 \; \text{for all }s,t\in S\). For a permutation $\sigma$ of $S$, define the rank cone
\(K_\sigma := \left\{ x\in X: x_{\sigma(1)}\leq\cdots\leq x_{\sigma(n)} \right\}\).
Every pair of vectors in $K_\sigma$ is comonotonic. The cones cover $X$ and
overlap at profiles with ties.

Let \(\succeq\) be a binary relation on \(X\), with symmetric and asymmetric
parts \(\sim\) and \(\succ\).

\begin{definition}
\label{def:capacity}
A \emph{capacity} on \(S\) is a set function
\(
    \nu:2^S\to[0,1]
\)
satisfying
\(
    \nu(\varnothing)=0,
    \;
    \nu(S)=1,
\)
and
\(
    A\subseteq B
    \;\Longrightarrow\;
    \nu(A)\leq\nu(B).
\)
The capacity is \emph{additive} if
\(
    \nu(A\cup B)=\nu(A)+\nu(B)
\)
for all disjoint \(A,B\subseteq S\).
\end{definition}

For \(x\in X\), choose a permutation \(\sigma\) such that
\(
    x_{\sigma(1)}\leq\cdots\leq x_{\sigma(n)}
\)
and define the upper events
\begin{equation}
    A_i^\sigma
    :=
    \{\sigma(i),\ldots,\sigma(n)\},
    \qquad
    i=1,\ldots,n,
    \label{eq:upper-events}
\end{equation}
with \(A_{n+1}^\sigma:=\varnothing\). The finite Choquet integral of \(x\)
with respect to \(\nu\) is
\begin{equation}
\begin{split}
    C_\nu(x)
    :=
    x_{\sigma(1)}
    +
    \sum_{i=2}^{n}
    \bigl(x_{\sigma(i)}-x_{\sigma(i-1)}\bigr)
    \nu(A_i^\sigma).
    \label{eq:choquet-integral}
\end{split}
\end{equation}
The value is independent of the ordering chosen at ties. On any fixed rank
cone, \(C_\nu\) is linear in \(x\). Although preferences are defined only on
\(X\), the same finite-sum formula canonically defines \(C_\nu(z)\) for every
\(z\in\mathbb R^S\). We use this extension only for algebraic manipulations in
the proofs. It satisfies
\(
    C_\nu(z+c\mathbf1)=C_\nu(z)+c
\)
for every \(c\in\mathbb R\), and
\(
    C_\nu(\lambda z)=\lambda C_\nu(z)
\)
for every \(\lambda\geq0\); all primitive preference comparisons remain on
\(X\).

\section{Representation}
\label{sec:representation}

Preferences and the behavioral axioms are stated on the full act domain
\(X=[0,1]^S\). The capacity schedule is indexed only by interior certainty-equivalent
levels $v\in(0,1)$; no capacities at $0$ or $1$ are needed. To obtain a
full-domain characterization nevertheless, we impose one boundary nondegeneracy
condition: every nontrivial boundary act must be indifferent to an interior
constant. Axioms \ref{ax:weak-order-continuity}--\ref{ax:strict-constants} already give every interior act a unique interior certainty equivalent; this condition ensures that the only acts with
endpoint certainty equivalents are $\mathbf0$ and $\mathbf1$. Thus every
nonendpoint act is evaluated by an identified member of the interior schedule
$\{\nu_v\}_{v\in(0,1)}$. The condition is a nondegeneracy requirement, not
a strict-monotonicity assumption: it rules out boundary acts whose certainty
equivalent coincides with one of the two endpoint constants.

The characterization rests on five behavioral axioms. The first three are
standard regularity and monotonicity requirements. The fourth replaces global
mixture-betweenness by its comonotonic counterpart, allowing the local state
weighting to vary across rank cones and, ultimately, to be represented by a
capacity. The fifth is the boundary nondegeneracy condition that makes the
interior capacity schedule sufficient to represent preferences on all of $X$.

\begin{axiom}[Continuous weak order]
\label{ax:weak-order-continuity}
The relation \(\succeq\) is complete and transitive. For every \(x\in X\),
the upper and lower contour sets
\(
    \{y\in X:y\succeq x\}
    \;\text{and}\;
    \{y\in X:x\succeq y\}
\)
are closed.
\end{axiom}

\begin{axiom}[Statewise monotonicity]
\label{ax:monotonicity}
If \(x_s\geq y_s\) for every \(s\in S\), then \(x\succeq y\).
\end{axiom}

\begin{axiom}[Strict ordering of constant acts]
\label{ax:strict-constants}
For all \(a,b\in I\),
\(
    a>b
    \;\Longrightarrow\;
    a\1\succ b\1.
\)
\end{axiom}

\begin{axiom}[Comonotonic mixture-betweenness]
\label{ax:comonotonic-betweenness}
For all comonotonic \(x,y\in X\) and every \(\lambda\in(0,1)\): (i) if \(x\succ y\), then
    \(
        x\succ
        \lambda x+(1-\lambda)y
        \succ y;
    \) and (ii) if \(x\sim y\), then
    \(
        \lambda x+(1-\lambda)y\sim x.
    \)

\end{axiom}

Axiom \ref{ax:comonotonic-betweenness} is ordinary
mixture-betweenness when the comparison is restricted to a rank cone. It is
strictly weaker than imposing mixture-betweenness globally, because acts in
different rank cones need not satisfy the axiom with one another.

\begin{axiom}[Boundary interiority]
\label{ax:boundary-interiority}
For every
\(x\in\partial X\setminus\{\mathbf0,\mathbf1\}\), there exists
\(v\in(0,1)\) such that
\(
    x\sim v\mathbf1.
\)
\end{axiom}

This is a substantive boundary nondegeneracy condition: every nontrivial
boundary act must have an interior certainty equivalent. It excludes models
in which such an act has an endpoint certainty equivalent. Together
with Axioms \ref{ax:weak-order-continuity}--\ref{ax:strict-constants}, the
condition ensures that every
$x\in X\setminus\{\mathbf0,\mathbf1\}$ has a unique certainty equivalent in
$(0,1)$. Hence the interior capacity schedule
$\{\nu_v:v\in(0,1)\}$ is sufficient to represent preferences on the full act
domain, without introducing separate endpoint capacities. In the subjective
expected utility benchmark, the condition is equivalent to full support of the
subjective probability. In the fixed-capacity Choquet benchmark, it requires
$0<\nu(A)<1$ for every nontrivial event
$\varnothing\subsetneq A\subsetneq S$, thereby excluding capacities that treat
a nontrivial event as null or certain.

\begin{definition}
\label{def:admissible-family}
A family $\{\nu_v:v\in(0,1)\}$ is \emph{admissible} if:
\begin{enumerate}[label=(\roman*),leftmargin=2.5em]
    \item every $\nu_v$ is a normalized monotone capacity;
    \item $v\mapsto\nu_v$ is continuous under
    \(
        \|\nu-\mu\|_\infty
        :=\max_{A\subseteq S}|\nu(A)-\mu(A)|;
    \)
    \item for every $x\in X\setminus\{\mathbf0,\mathbf1\}$, the equation
    \(
        C_{\nu_v}(x)=v
    \)
    has a unique solution $v\in(0,1)$.
\end{enumerate}
\end{definition}
The unique-root requirement is the representation-side coherence counterpart
of the behavioral fact that every act has a unique certainty equivalent. It
prevents a proposed capacity schedule from assigning two distinct
self-consistent certainty-equivalent levels to the same act. Importantly, this condition is
not an additional behavioral assumption in the forward direction of the
theorem: it is derived from the axioms and uniqueness of certainty
equivalents. The converse uses it to ensure that the implicit equation defines
a single-valued preference.

Admissibility is not a knife-edge restriction. The following remark provides a
simple sufficient condition and nonconstant examples.
\begin{remark}
\label{rem:admissibility-sufficient}
Suppose that the capacity
schedule is Lipschitz under the sup norm with constant $L<1$:
\(\|\nu_w-\nu_v\|_\infty \leq L|w-v|\).
The Lipschitz schedule has unique endpoint limits $\nu_0$ and $\nu_1$.
Suppose moreover that
\(\nu_0(A)>0 \;\text{for every nonempty }A\subseteq S, \; \nu_1(A)<1 \;\text{for every proper }A\subsetneq S\).
Then, for every $x\in X\setminus\{\mathbf0,\mathbf1\}$,
\(
    C_{\nu_0}(x)>0
    \; \text{and} \;
    C_{\nu_1}(x)<1.
\)
Moreover, for $w>v$,
\begin{align*}
    &[C_{\nu_w}(x)-w]-[C_{\nu_v}(x)-v]\\
    &\qquad\leq
    \operatorname{osc}(x)\|\nu_w-\nu_v\|_\infty-(w-v)
    \leq -(1-L)(w-v)<0,
\end{align*}
where $\operatorname{osc}(x)=\max_s x_s-\min_s x_s$. Hence $v\mapsto C_{\nu_v}(x)-v$ is strictly decreasing. The endpoint
inequalities give opposite signs near $0$ and $1$ for every
$x\notin\{\mathbf0,\mathbf1\}$, so each such act has exactly one root in
$(0,1)$. Interior constant acts $c\mathbf1$, $c\in(0,1)$, have the unique
root $v=c$ automatically.

For an explicit nonconstant class, take two distinct capacities
$\nu^0$ and $\nu^1$ satisfying
\(0<\nu^k(A)<1 \;\text{for every } \varnothing\subsetneq A\subsetneq S, \; k=0,1\),
a nonconstant Lipschitz function $g:[0,1]\to[0,1]$, and
$\varepsilon\in(0,1]$ satisfying
$\varepsilon\operatorname{Lip}(g)<1$. Then
\(\nu_v = [1-\varepsilon g(v)]\nu^0 +\varepsilon g(v)\nu^1, \; v\in(0,1)\),
is an admissible nonconstant family. If $\nu^0$ and $\nu^1$ are also convex,
then every $\nu_v$ is convex.\hfill \qed
\end{remark}

The preceding axioms restrict preferences locally within comonotonic regions while allowing the evaluation of uncertainty to vary across indifference levels. The main result shows that these restrictions are exactly captured by a continuous family of capacities indexed by the decision maker's certainty-equivalent level. Thus, although the weighting of events may change with the certainty-equivalent level, all nonendpoint acts sharing the same interior certainty-equivalent level are evaluated using the same capacity. The resulting representation provides a disciplined form of utility-level-dependent ambiguity while retaining the Choquet structure within each level.

\begin{theorem}
\label{thm:implicit-choquet}
On the full act domain $X=[0,1]^S$, the following are equivalent:
\begin{enumerate}[label=(\roman*),leftmargin=2.5em]
    \item Preferences satisfy Axioms
    \ref{ax:weak-order-continuity}--\ref{ax:boundary-interiority}.

    \item There exists a unique admissible family of capacities
    $\{\nu_v:v\in(0,1)\}$ such that, defining
    \(
        V(\mathbf0):=0,
        \;
        V(\mathbf1):=1,
    \)
    and, for every $x\in X\setminus\{\mathbf0,\mathbf1\}$, letting $V(x)$ be
    the unique solution in $(0,1)$ of \(C_{\nu_v}(x)=v\),
    the map $V:X\to[0,1]$ represents preferences:
    \begin{equation}
        x\succeq y
        \quad\text{if and only if}\quad
        V(x)\geq V(y),
        \quad x,y\in X.
        \label{eq:main-value-representation}
    \end{equation}
\end{enumerate}
\end{theorem}

By construction, every \(x\in X\setminus\{\mathbf0,\mathbf1\}\) satisfies
\(C_{\nu_{V(x)}}(x)=V(x)\). For nonendpoint acts, admissibility additionally
implies the single-crossing orientation of the implicit equation; together
with normalization at the endpoints, this yields, for every $x\in X$ and
$v\in(0,1)$,
\begin{equation}
x\succeq v\mathbf1
\quad\Longleftrightarrow\quad
V(x)\geq v
\quad\Longleftrightarrow\quad
C_{\nu_v}(x)\geq v.
\label{eq:main-threshold}
\end{equation}
The corresponding strict and equality relations follow analogously.

Relative to fixed-capacity Choquet expected utility, the theorem permits the
capacity to vary with the certainty-equivalent level while imposing a strong
cross-act discipline: the capacity is indexed only by the scalar $V(x)$, so
all nonendpoint acts on the same interior indifference surface share the same
event weighting.
Equation \eqref{eq:main-threshold} gives the corresponding local orientation
of preferences around every interior certainty-equivalent level. The boundary
nondegeneracy axiom makes this interior schedule sufficient for the full act
domain.

\paragraph{Proof sketch.}
For the forward direction, Axioms
\ref{ax:weak-order-continuity}--\ref{ax:strict-constants} first deliver a
continuous certainty-equivalent representation $V$. Fix a rank cone
$K_\sigma$. Because all acts in $K_\sigma$ are comonotonic, Axiom
\ref{ax:comonotonic-betweenness} becomes ordinary mixture-betweenness on that
cone. Axiom \ref{ax:boundary-interiority}, together with strict ordering of
constants, makes $\mathbf0$ and $\mathbf1$ respectively strict worst and strict
best acts. \citet{Payro2025}'s mixture-space representation can therefore be applied
locally: for each $v\in(0,1)$ there is a continuous vector $q^\sigma(v)$, normalized so
that its coordinates sum to one, such that
\(
    V(x)=v
    \;\text{if and only if}\;
    q^\sigma(v)\cdot x=v,
    \; x\in K_\sigma.
\)
Comonotonic betweenness determines the preferred side of this hyperplane, and
statewise monotonicity implies $q^\sigma(v)\in\Delta(S)$. Up to this point the
argument uses established implicit-linear representation machinery. The key
construction begins with the cross-cone consistency problem: for a fixed
certainty-equivalent level $v$, the local objects $q^\sigma(v)$ depend on the ranking
$\sigma$, and it is not automatic that they are generated by one event
capacity. For an event $A$, binary acts that take utility value $v+a$ on $A$ and
$v-b$ on $A^c$ belong to every rank cone in which $A$ is an upper event. Since all such
cones must represent the same primitive comparison, the aggregate weight
$\sum_{s\in A}q_s^\sigma(v)$ must be independent of the chosen cone. This
binary-act consistency argument is what permits the level-by-level definition
of a single monotone capacity $\nu_v$. Along each rank cone its Choquet
marginal weights coincide with $q^\sigma(v)$, so
\(
    C_{\nu_v}(x)=q^\sigma(v)\cdot x,
\)
which yields the threshold representation
\eqref{eq:main-threshold}. Thus the proof does not require
a new separation theorem or a new property of the Choquet integral; its main
additional step is the cross-rank event consistency that assembles level-indexed local
probability vectors into $\nu_v$. The same binary-event threshold comparisons identify each
$\nu_v$ uniquely. Moreover, any rival admissible representing family has the
same normalized certainty-equivalent map and therefore the same level-by-level
threshold comparisons, so this event-wise identification yields uniqueness of
the entire family. Continuity of the local vectors $q^\sigma(v)$ implies
continuity of $v\mapsto\nu_v$. The unique certainty equivalent of each act then
becomes the unique solution of $C_{\nu_v}(x)=v$, establishing admissibility.

For the converse, start from an admissible family $\{\nu_v\}$, define the
endpoint values directly, and use the unique interior root for every other act.
Admissibility implies a single-crossing property for $C_{\nu_v}(x)-v$ around
that root, including for nonconstant boundary acts. This yields the threshold
representation and statewise monotonicity. The root map is continuous on the
full act domain, including at $\mathbf0$ and $\mathbf1$, while constants satisfy
$V(c\mathbf1)=c$. Finally, fixed-$v$ Choquet affinity on comonotonic pairs
places every strict mixture between the certainty-equivalent levels of its
endpoints and preserves equality for indifferent endpoints. Hence all five
axioms follow on $X$. The detailed construction and both directions are given
in Appendices \ref{app:auxiliary} and \ref{app:representation-proofs}.

\section{Benchmark restrictions}
\label{sec:special-cases}

The representation separates two logically distinct departures from
full-support subjective expected utility within the maintained class:
nonadditivity across states and dependence of the
local state weighting on the indifference level. The benchmark restrictions
below now apply to the represented preference on the full domain $X$.

To separate rank dependence from utility-level dependence, first ask what
changes when mixture-betweenness is imposed globally rather than only within
comonotonic regions.

\begin{proposition}
\label{prop:global-betweenness}
Within the class characterized above, the following are equivalent:
\begin{enumerate}[label=(\roman*),leftmargin=2.5em]
    \item mixture-betweenness holds for all pairs of acts, not only
    comonotonic pairs;
    \item \(\nu_v\) is additive for every \(v\in(0,1)\).
\end{enumerate}
In that case there is a continuous family
\(p(v)\in\Delta(S)\) such that
\(V(x)=v \;\text{if and only if}\; p(v)\cdot x=v\).
\end{proposition}

Thus global betweenness eliminates nonadditivity: the state weighting may
still vary with \(v\), but at each certainty-equivalent level it is an ordinary probability.

The proposition identifies the level-dependent additive-probability benchmark studied
by the mixture-betweenness literature. In particular, the finite-state
Anscombe--Aumann specialization of \citet{Payro2025} takes the form
of an implicit expected utility whose state probabilities may depend on the
implicit utility level.

The complementary restriction asks what happens when comparisons remain
comonotonic but betweenness is strengthened to independence. For clarity, say
that preferences satisfy \emph{comonotonic independence} if, for every
pairwise comonotonic $x,y,z\in X$ and every $\alpha\in(0,1)$,
\(x\succeq y \;\text{if and only if}\; \alpha x+(1-\alpha)z \succeq \alpha y+(1-\alpha)z\).
This is the standard independence restriction applied only within a common
comonotonic region.

\begin{proposition}
\label{prop:fixed-capacity}
Within the class characterized by Theorem \ref{thm:implicit-choquet},
strengthening comonotonic mixture-betweenness to comonotonic independence is
equivalent to a constant capacity family:
\(\nu_v=\nu \; \text{for all }v\in(0,1)\),
with
\(V(x)=C_\nu(x), \; x\in X\).
The fixed capacity necessarily satisfies
\(
    0<\nu(A)<1
\)
for every nontrivial event
\(\varnothing\subsetneq A\subsetneq S\).
\end{proposition}

Thus comonotonic independence eliminates cross-level variation while
retaining the possibility of nonadditive state weighting. This recovers the
maintained nondegenerate fixed-capacity CEU benchmark. Because Axiom
\ref{ax:boundary-interiority} is maintained, this fixed-capacity CEU benchmark
excludes capacities with null or certain nontrivial events; equivalently,
$0<\nu(A)<1$ for every
$\varnothing\subsetneq A\subsetneq S$.

The preceding two restrictions eliminate the two departures from subjective
expected utility separately. Imposing both therefore recovers a single fixed
additive probability.

\begin{corollary}
\label{cor:seu-intersection}
Within the class characterized by Theorem~\ref{thm:implicit-choquet}, the
following are equivalent:
\begin{enumerate}[label=(\roman*),leftmargin=2.5em]
    \item preferences satisfy global mixture-betweenness and comonotonic
    independence;
    \item the representing family is both additive and independent of the
    certainty-equivalent level: there exists a single full-support probability
    $p\in\Delta(S)$ such that
    \(\nu_v(A)=\sum_{s\in A}p_s\) for every $A\subseteq S$ and
    $v\in(0,1)$;
    \item preferences admit the full-support subjective expected utility representation
    \(V(x)=p\cdot x \; \text{for all }x\in X\)
    for a single full-support $p\in\Delta(S)$.
\end{enumerate}
\end{corollary}

Full-support subjective expected utility is therefore exactly the intersection
of the two benchmark restrictions within the maintained class: additivity at
every level and invariance of the weighting object across levels. The full-support
qualification is exactly the boundary implication of Axiom
\ref{ax:boundary-interiority}.

\begin{figure}[htbp]
\centering
\begin{tabular}{>{\raggedright\arraybackslash}p{0.28\textwidth}
                >{\centering\arraybackslash}p{0.27\textwidth}
                >{\centering\arraybackslash}p{0.30\textwidth}}
\toprule
& \textbf{Fixed across CE levels}
& \textbf{May vary with \(v\)} \\
\midrule
\textbf{Additive at every level}
& Full-support SEU: \(p\)
& Implicit additive utility: \(p(v)\) \\
\addlinespace
\textbf{Nonadditive at some level}
& Nonadditive fixed-capacity CEU: \(\nu\)
& Nonadditive implicit Choquet utility: \(\nu_v\) \\
\bottomrule
\end{tabular}
\caption{Four benchmark classes within the maintained nondegenerate class}
\label{fig:four-classes}
\end{figure}

The vertical move from the top to the bottom row permits genuine
nonadditivity by weakening the global mixture restriction to a comonotonic one.
The horizontal move from the left to the right column permits the weighting
object to vary across indifference levels by weakening independence to
betweenness. The general representation encompasses all four cells: global
mixture-betweenness restricts it to the top row, comonotonic independence
restricts it to the left column, and imposing neither restriction permits the
genuinely nonadditive, level-varying bottom-right case.

\section{Ambiguity implications}
\label{sec:ambiguity-analysis}

The representation permits the ambiguity evaluation to vary with the decision
maker's certainty-equivalent level while remaining common to all nonendpoint acts on a given
indifference surface. This section develops the main implications of that
structure: local ambiguity attitudes, certainty-equivalent-indexed multiple
priors, revealed-preference identification, the empirical content and
falsifiability of the representation, cross-level changes in event weighting,
certainty-level shifts and common translations, and the qualifications needed
for interpretation.

\subsection{Ambiguity attitudes at a certainty-equivalent level}
\label{sec:local-ambiguity}

The representation gives each interior indifference level its own Choquet
functional. This section asks what familiar ambiguity-attitude restrictions
mean locally. The first notion is the level-specific analogue of uncertainty
aversion: randomizing between two acts on the same indifference surface should
not reduce welfare.

\begin{definition}
\label{def:local-ua}
Fix $v\in(0,1)$. Preferences satisfy \emph{local uncertainty aversion at
$v$} if, for every $x,y\in X$ satisfying
\(
    x\sim v\mathbf1
    \;\text{and}\;
    y\sim v\mathbf1,
\)
and every $\lambda\in[0,1]$, \(  \lambda x+(1-\lambda)y
    \succeq
    v\mathbf1\). Local uncertainty neutrality at $v$ requires indifference; local uncertainty seeking reverses the weak inequality.
\end{definition}

The next result gives local uncertainty aversion a familiar Choquet
interpretation. Recall that a capacity is convex (supermodular) if
\(\nu(A\cup B)+\nu(A\cap B) \geq \nu(A)+\nu(B) \; \text{for all }A,B\subseteq S\).

\begin{proposition}
\label{prop:local-ua}
Fix $v\in(0,1)$. Under the representation in Theorem
\ref{thm:implicit-choquet}, the following are equivalent:
\begin{enumerate}[label=(\roman*),leftmargin=2.5em]
    \item preferences satisfy local uncertainty aversion at $v$;
    \item the functional $C_{\nu_v}$ is concave on $X$;
    \item $\nu_v$ is a convex capacity.
\end{enumerate}
Likewise, local uncertainty neutrality at $v$ is equivalent to additivity of
$\nu_v$.
\end{proposition}

The proposition gives the behavioral condition a precise local meaning:
the theorem identifies the event weighting $\nu_v$, while convexity determines
when that weighting has the standard ambiguity-averse shape. The condition may
therefore hold at some certainty-equivalent levels but not others.

\subsection{Certainty-equivalent-indexed multiple priors}

For a convex capacity $\nu$, define its core by
\(\operatorname{core}(\nu) := \left\{ p\in\Delta(S): p(A)\geq\nu(A) \text{ for every }A\subseteq S \right\}\).
A convex capacity has a nonempty core, and its Choquet integral is the lower
envelope of expectations over that core; this is the familiar connection
between ambiguity-averse Choquet and multiple-priors representations
\citep{Schmeidler1989,GilboaSchmeidler1989}.

Once the local capacities are convex, it is natural to ask whether the
representation can be written in the familiar multiple-priors form without
losing its utility-level dependence.

\begin{corollary}
\label{cor:implicit-multiple-priors}
Suppose local uncertainty aversion holds at every interior certainty-equivalent level. For
$v\in(0,1)$ define
\(
    \mathcal P_v:=\operatorname{core}(\nu_v).
\)
Then $\mathcal P_v$ is nonempty and compact, and
\(x\succeq v\mathbf1\) if and only if
\(\min_{p\in\mathcal P_v}p\cdot x\geq v\). For every
$x\in X\setminus\{\mathbf0,\mathbf1\}$,
\(V(x)=\min_{p\in\mathcal P_{V(x)}}p\cdot x\).
\end{corollary}

Thus local uncertainty aversion converts the capacity schedule into a
schedule of locally relevant prior sets. Unlike maxmin expected utility, the
set of priors is local to the certainty-equivalent level generated by the act
rather than a fixed global object. This multiple-priors interpretation is
available precisely in the convex-capacity subclass.

\subsection{Revealed-preference identification and elicitation}
\label{sec:revealed}

The local capacity is identified from ordinary preference comparisons rather
than from an exogenous belief report. Proposition
\ref{prop:revealed-capacity} shows that binary event comparisons reveal the
cutoff associated with $\nu_v(A)$. Whenever that cutoff is interior, an
indifference
\((v+a)\mathbf1_A+(v-b)\mathbf1_{A^c} \sim v\mathbf1\)
reveals
\(\nu_v(A)=\frac{b}{a+b}\).
If $\nu_v(A)=0$ or $1$, no interior gain--loss indifference exists.
Although normalization fixes $\nu_v(\varnothing)=0$ and $\nu_v(S)=1$, an
admissible level-dependent schedule may assign capacity $0$ or $1$ to a
nontrivial event at particular certainty-equivalent levels. These boundary
values are instead revealed by one-sided comparisons: $\nu_v(A)=0$ if and
only if increasing utility only on $A$ leaves the act indifferent to
$v\mathbf1$, while $\nu_v(A)=1$ if and only if decreasing utility only on
$A^c$ leaves it indifferent to $v\mathbf1$. The primitive observable is the
binary gain--loss cutoff itself. Conditional on the common utility calibration
and on the behavioral restrictions supporting the capacity representation,
the comparison at each elicited certainty-equivalent level point-identifies
$\nu_v(A)$. The full preference relation identifies the entire profile
$v\mapsto\nu_v(A)$ in principle; a finite experiment traces that profile only
at the elicited levels unless an interpolation or functional-form restriction
is added. If the maintained axioms fail, the cutoffs remain observable but need
not assemble into a globally coherent capacity schedule.

The representation suggests a direct elicitation design. Hold the event
structure fixed, vary the calibrated reference certainty level, and repeatedly
locate the binary gain--loss cutoff; when the cutoff is interior, this amounts
to finding the utility tradeoff that makes the binary ambiguous bet indifferent
to the reference level. The design can hold the event and local gain--loss
structure fixed while shifting the reference certainty-equivalent level, which
helps distinguish level dependence from ordinary stake variation. It does not,
by itself, eliminate every competing channel: stake changes can also alter the
overall attractiveness and dispersion of an act, and reference dependence or
probability sensitivity may interact with the elicitation. Conditional on the
common vNM utility scale and the maintained behavioral model, however, the
design point-identifies $\nu_v(A)$ at each elicited level and traces its
cross-level pattern from binary threshold comparisons. Estimating a continuous
profile from finitely many elicited levels additionally requires interpolation
or functional-form structure.

The quantities $v$, $a$, and $b$ are utility-scale levels rather than raw monetary
amounts. A literal monetary implementation therefore requires calibration of
the common vNM utility index. In an Anscombe--Aumann experiment this can be
done directly: after normalizing a worst and a best lottery to utilities $0$
and $1$, objective mixtures of those lotteries implement any desired utility
level by their mixture probability. The elicitation can therefore vary the
reference certainty level and the local gain--loss increments in utility units
without assuming linear utility for money. The resulting structural
identification is conditional on both this calibration and the behavioral
restrictions of the representation; it is not a model-free elicitation of raw
monetary ambiguity weights.

\subsection{Empirical content and falsifiability}
\label{sec:falsifiability}

Allowing the capacity to vary with the certainty-equivalent level makes the
implicit Choquet model more flexible than fixed-capacity CEU, but it does not
make the model unrestricted. It is useful to distinguish restrictions that can
reject the implicit Choquet representation itself from restrictions that reject
only particular subclasses within it. The characterization theorem already
gives direct model-level tests. In particular, violations of statewise
monotonicity, strict ordering of constants, or the maintained boundary
interiority condition reject the represented class. Most distinctively, if
$x$ and $y$ are comonotonic and indifferent, comonotonic
mixture-betweenness requires every feasible mixture
$\lambda x+(1-\lambda)y$ to remain indifferent to them. A strict preference
for or against such a mixture directly violates the representation.

The capacity structure supplies additional overidentifying restrictions at
each fixed certainty-equivalent level. Binary comparisons can be used to
recover event weights $\nu_v(A)$, but those weights cannot be chosen
independently event by event: together they must form one normalized monotone
capacity. Hence they must satisfy
$\nu_v(\varnothing)=0$, $\nu_v(S)=1$, and
$A\subseteq B\Rightarrow \nu_v(A)\leq\nu_v(B)$. Once these event weights have
been elicited, the same capacity must also classify additional acts not used in the elicitation according to
\(
    x\succeq v\mathbf1
    \Longleftrightarrow
    C_{\nu_v}(x)\geq v.
\)
Failure of these within-level predictions rejects the implicit Choquet model,
rather than merely a special case of it.

A particularly sharp test targets the cross-rank consistency restriction behind the main
representation. Suppose local comparisons in two rank cones $K_\sigma$ and
$K_\tau$ recover the corresponding probability vectors $q^\sigma(v)$ and
$q^\tau(v)$. Whenever an event $A$ is an upper set in both rankings, Lemma
\ref{lem:binary-gluing} requires
\(
    \sum_{s\in A}q_s^\sigma(v)
    =
    \sum_{s\in A}q_s^\tau(v).
\)
If the same event receives different aggregate weights across such rank cones,
the local linear evaluations cannot be glued into a single capacity $\nu_v$;
this rejects precisely the implicit Choquet restriction while remaining
compatible with a more permissive rank-cone-specific implicit model.

The schedule also imposes cross-level coherence. If an act has certainty
equivalent $r=V(x)$, Lemma \ref{lem:single-crossing} requires
\(
    C_{\nu_v}(x)-v>0 \text{ for } v<r;
    \;
    C_{\nu_r}(x)-r=0;
    \;
    C_{\nu_v}(x)-v<0 \text{ for } v>r.
\)
Thus independently elicited capacities at different levels must orient the
same act consistently around its certainty equivalent. Continuity of
$v\mapsto\nu_v$ by itself has relatively little finite-sample bite, since a
finite set of admissible observations can often be joined by a continuous
schedule; the stronger empirical discipline comes from the many within-level
and cross-rank restrictions imposed by a single capacity at each $v$.

By contrast, several natural violations reject only benchmark subclasses.
Finding $\nu_v(A)\neq\nu_w(A)$ for some $v\neq w$ rejects the fixed-capacity
CEU restriction, not the implicit Choquet model. Failure of additivity at a
fixed level rejects the implicit additive subclass, while failure of convexity
rejects the local uncertainty-aversion subclass. The same elicited event
weights can therefore be used both to test the general representation and to
discriminate among its benchmark restrictions. The model is consequently flexible in how ambiguity weighting may change
across certainty-equivalent levels, while placing strong overidentifying
restrictions on choice within any given level.

\subsection{Cross-level changes in ambiguity weighting}
\label{sec:endogenous-ambiguity}

The representation allows two distinct kinds of variation. At a fixed $v$,
the nonadditivity of $\nu_v$, when present, describes the local departure from additive
event weighting. Across levels, changes in $v\mapsto\nu_v$ describe how the local
evaluation of the same event changes with the certainty-equivalent level.
The binary-bet formula in Proposition \ref{prop:revealed-capacity} makes the
second margin behaviorally testable through observable cutoff comparisons.

For a fixed event $A\subseteq S$, the map
\(
    v\longmapsto \nu_v(A)
\)
records how the local evaluation of that event changes across
certainty-equivalent levels. Through the binary-event comparisons in
Proposition \ref{prop:revealed-capacity}, this map has a direct
revealed-preference interpretation: it traces how willingness to accept
exposure to the same event changes as the reference certainty-equivalent level
changes. The index $v$ is endogenous---it is the certainty equivalent of the
act under evaluation---and is not itself an exogenous wealth or consumption
state. A change in a single event weight need not
have an unambiguous global interpretation. In the convex-capacity subclass,
however, a common movement across all events admits a stronger set-of-priors
interpretation.

\begin{definition}
\label{def:core-expansion}
Suppose $\nu_v$ and $\nu_w$ are convex, and let
$\mathcal P_z:=\operatorname{core}(\nu_z)$. For $w>v$, the local core
\emph{expands from $v$ to $w$} if \(   \mathcal P_v\subseteq\mathcal P_w\). It contracts if $\mathcal P_w\subseteq\mathcal P_v$.
\end{definition}

The definition is most useful if core inclusion can be read directly from the
capacity schedule and from local act evaluations. The next proposition
establishes exactly these equivalences.

\begin{proposition}
\label{prop:core-expansion}
Suppose $\nu_v$ and $\nu_w$ are convex. The following are equivalent:
\begin{enumerate}[label=(\roman*),leftmargin=2.5em]
    \item $\nu_w(A)\leq\nu_v(A)$ for every event $A$;
    \item $\mathcal P_v\subseteq\mathcal P_w$, where
    $\mathcal P_z=\operatorname{core}(\nu_z)$;
    \item $C_{\nu_w}(x)\leq C_{\nu_v}(x)$ for every $x\in X$.
\end{enumerate}
\end{proposition}

The proposition therefore gives three equivalent descriptions of the same
cross-level movement: event weights fall, the core expands, and every local
Choquet evaluation weakly falls.

The ordering in Proposition \ref{prop:core-expansion} is intrinsically
nonadditive. If both capacities are additive, say $\nu_z(A)=p_z(A)$, and
$p_w(A)\leq p_v(A)$ for every event, then applying the inequality to $A^c$
gives the reverse inequality for $A$. Hence $p_w=p_v$. A nontrivial uniform
expansion or contraction of the local core therefore cannot be mimicked merely
by changing an additive prior.

This distinction helps clarify the interpretation of utility-level dependence.
The model treats $\nu_v$ as a reduced-form ambiguity weighting and does not, by
itself, separate perceived ambiguity from attitude toward it. In the convex
subclass, monotone core expansion nevertheless has a transparent structural
meaning: the collection of locally relevant probability models becomes larger.
Section~\ref{subsec:related-attitudes} discusses the corresponding
interpretive literature.

\subsection{Translation invariance and background shifts}
\label{sec:background}

A central implication of utility-level-dependent ambiguity is that adding the
same certain amount to every state can change how the remaining uncertainty is
evaluated. Fixed-capacity Choquet expected utility rules this out because one
capacity applies at every level. The next result shows that this property
exactly characterizes the fixed-capacity benchmark within the present model.

\begin{definition}
\label{def:translation-invariance}
The certainty equivalent is \emph{translation invariant} if, whenever
$x,x+c\mathbf1\in X$, \(  V(x+c\mathbf1)=V(x)+c\).

\end{definition}

\begin{proposition}
\label{prop:translation-fixed}
Within the class characterized by Theorem \ref{thm:implicit-choquet}, the
following are equivalent:
\begin{enumerate}[label=(\roman*),leftmargin=2.5em]
    \item $V$ satisfies certainty translation invariance;
    \item $\nu_v=\nu_w$ for all interior $v,w$;
    \item there exists a fixed capacity $\nu$ such that
    $V(x)=C_\nu(x)$ for every $x\in X$.
\end{enumerate}
\end{proposition}

This equivalence makes sensitivity to common certainty translations
diagnostic: any nonconstant capacity schedule necessarily violates certainty
translation invariance. The result supplies a disciplined connection to
background shifts, but the index $v$ remains the act's endogenous certainty
equivalent rather than an independently specified wealth variable.

The result also connects to work on how ambiguity attitudes vary with the
decision maker's position and to models that impose stronger certainty-shift
invariance. Those comparisons---including wealth effects, constant absolute
uncertainty aversion, recursive translation invariance, and cognitive capacity
choice---are developed in Section \ref{sec:related}.

\subsection{Boundary and interpretation}
\label{sec:boundary-interpretation}

\paragraph{Endpoint identification.}
Axiom \ref{ax:boundary-interiority} permits a full-domain characterization
without requiring capacities at certainty-equivalent levels $0$ and $1$. Every
$x\in X\setminus\{\mathbf0,\mathbf1\}$ has an interior certainty equivalent,
so its value is determined by the identified schedule $\{\nu_v\}_{v\in(0,1)}$.
The endpoint acts are represented directly by
$V(\mathbf0)=0$ and $V(\mathbf1)=1$. Thus the theorem is a full-domain characterization even though no endpoint
capacities $\nu_0$ and $\nu_1$ are required by the representation. Strong statewise monotonicity would imply the
same boundary-interiority property, but the maintained axiom imposes only the
weaker nondegeneracy condition needed for the representation. In particular,
it requires each nontrivial boundary act to have an interior certainty
equivalent without requiring every coordinatewise improvement to be strictly
preferred. In the SEU benchmark this is exactly full support; in the fixed-capacity
CEU benchmark it rules out a nontrivial event receiving capacity zero or one.

\paragraph{State-dependent affine utilities.}
If an act is reduced using genuinely different state-dependent affine indices
$u_s$, then comparisons such as $u_s(f(s))>u_t(f(t))$ and hence the rank cones
depend on the relative cardinal normalization of $u_s$ and $u_t$. The theorem
therefore requires either a common affine utility over lotteries or an
independently justified interstate cardinal scale. Once such a scale is fixed,
the same proof applies to $x_s=u_s(f(s))$ provided the induced utility domain
retains the full-domain richness and mixture structure used here, including common
constant profiles and the required rank cones. Without this additional domain richness,
the present theorem does not automatically extend to arbitrary state-dependent
utility specifications. Without interstate comparability, a rank-dependent
capacity interpretation is not invariant to admissible state-by-state
rescalings.

\paragraph{Affine reindexing.}
The normalization to $[0,1]$ fixes the common consequence scale once the worst
and best benchmark lotteries are chosen. Under an alternative positive affine
labeling $v'=a+bv$, $b>0$, the normalized range $[0,1]$ is mapped to
$[a,a+b]$, and on that transformed range the same preference reindexes the
capacity schedule as $\nu'_{v'}=\nu_{(v'-a)/b}$. Renormalizing the
transformed scale back to $[0,1]$ restores the original indexing convention.
Thus constancy of the schedule, event-wise orderings, and convexity are
invariant to positive affine relabeling, whereas numerical slopes with respect
to the level index are meaningful only relative to the chosen calibration.

\paragraph{Capacity interpretation.}
A capacity should not automatically be interpreted as a literal
subjective probability. At the level of the representation theorem, $\nu_v$
is a local event-weighting object, possibly nonadditive, identified by preferences. The
terminology of ambiguity is strongest when additional restrictions give it the
standard ambiguity-averse interpretation. In particular, when $\nu_v$ is
convex, the core representation provides the stronger structural
interpretation
\(
    C_{\nu_v}(x)=\min_{p\in\mathcal P_v}p\cdot x.
\)

\section{Related literature}
\label{sec:related}

The paper sits at the intersection of two lines of axiomatic decision theory
that have largely developed separately. One allows the local evaluation of an
act to depend on an implicitly determined utility or indifference level; the
other relaxes additivity across states through rank-dependent or Choquet-type
evaluation. The present representation combines both margins: event weighting may be
nonadditive, and the relevant capacity is indexed by the overall
certainty-equivalent level. Related work on ambiguity attitudes and multiple
priors helps clarify the economic interpretation of this certainty-equivalent-indexed capacity
schedule.

\subsection{Implicit and betweenness representations under uncertainty}
\label{subsec:related-implicit}

The first strand studies implicit representations generated by weakenings of
independence. Under risk, \citet{Dekel1986} shows how betweenness permits the
linear evaluation supporting an indifference surface to vary with the utility
level. \citet{ChewEpstein1989}, together with the correction in
\citet{ChewEpsteinWakker1993}, provide an especially close precedent: their
implicit rank-linear framework unifies betweenness and rank dependence on
objective lotteries, and their implicit rank-dependent expected utility
specialization permits the probability transformation itself to depend on the
utility level. The correction establishes that the auxiliary Theorem A used in
the original proof is false while preserving Theorem 2 with minor
modifications. Thus neither utility-level-dependent rank weighting nor the
combination of implicit and rank-dependent evaluation is claimed as new here.
\citet{Payro2025} develops a general mixture-betweenness representation on
mixture spaces and applies it to uncertainty. In the present finite-state Anscombe--Aumann environment, global
mixture-betweenness yields the additive special case
\(
    V(x)=p(V(x))\cdot x.
\)
The current paper preserves this indifference-level dependence but restricts
the mixture requirement to comonotonic acts, which allows the local state
weights to vary across rank cones and, after imposing cross-rank consistency,
produces a capacity rather than an additive probability.

Several papers on acts provide closely related implicit structures.
\citet{GrantKajiiPolak2000} introduce weak decomposability and obtain an
implicit state-additive representation under uncertainty. Their result is an
important predecessor of the additive branch of the present model: local
evaluations may depend on the value assigned to the act, but event aggregation
remains additive. \citet{GrantRoordaYang2025} develop Expected Balanced
Uncertain Utility, in which acts are ranked implicitly through a balancing
value and a balancing-value-dependent outcome-set utility. Importantly, their
implicit Hurwicz example already allows the degree of ambiguity aversion to
depend on the indifference class. Thus indifference-class dependence of
ambiguity attitudes is explicitly conceded as a precedent rather than claimed
as the novelty here. The additional object identified by the present theorem
is a single \emph{event capacity} \(\nu_v\) at each
certainty-equivalent level, constructed from rank-specific local probability
vectors and behaviorally identified by binary event threshold comparisons.
EBUU is organized around a prior on decomposable events and a balancing-value-
dependent utility of outcome sets, whereas the present representation is
organized around a capacity on subjective events at each certainty-equivalent
level. The paper does not establish, and does not rely on, a general nesting
relation between these two representation classes; such a relation is not
immediate from their different primitive objects.
\citet{RoordaJoosten2026}
ask when a preference ordering on finite-state monetary acts admits an additive
state-dependent utility representation per indifference curve and develop
additive sliding-balance representations. Relative to this strand, the present
paper's main additional margin is therefore the capacity indexed by the
indifference level and its event-wise identification.

\subsection{Choquet, cumulative, and consequence-dependent state weighting}
\label{subsec:related-choquet}

The second strand concerns nonadditive or consequence-sensitive state
weighting. \citet{Schmeidler1989} obtains Choquet expected utility from
comonotonic independence, with one fixed capacity applying to every act.
Because the present paper is finite-state, \citet{ChewKarni1994} is also a
natural benchmark: they axiomatize subjective and Choquet expected utility on
a finite state space through act-independence and comonotonic
act-independence. Section~\ref{sec:special-cases} shows that strengthening the
present model to comonotonic independence likewise eliminates the
certainty-equivalent-level dependence and recovers a fixed capacity.

\citet{ChewWakker1996} weaken the sure-thing principle on comonotonic acts and
characterize cumulative utility, a general rank-dependent form that extends
the Choquet integral and permits event evaluation to depend on consequence
levels. This is close to the present paper on the nonadditivity dimension but
different in the source of dependence: cumulative utility links event weights
to the consequence being evaluated, whereas here all nonendpoint acts on the
same interior indifference surface share one capacity $\nu_v$, indexed by the
overall value $v$. \citet{Hazen1987} is another close antecedent. His subjectively weighted
linear utility model allows assessed subjective probabilities to depend on the
consequences associated with events, while retaining additivity for disjoint
events. The present representation instead allows genuinely nonadditive local
state weighting and disciplines its act dependence through the scalar
certainty-equivalent index.

More recently, \citet{BaillonBleichrodtLiWakker2025} develop Source Theory,
which represents ambiguity through Savage subjective probabilities combined
with source-dependent nonlinear weighting functions. Their central variation
is across sources of uncertainty: the same subjective probability can receive
different decision weight under different sources. This is complementary to
the present construction. Here the source is held fixed and the weighting
object is instead a capacity on subjective events that may vary with the act's
certainty-equivalent level. The two approaches therefore index non-EU event
weighting along different dimensions, and no general nesting relation is
claimed.

The relation between these representation classes is therefore not an
immediate nesting statement. Cumulative utility uses an outcome-dependent
weighting $W(z,A)$, so an event's incremental weight may vary with the
particular consequence level $z$ inside the same act. Here, by contrast, a
single capacity $\nu_{V(x)}$ applies to every step of the Choquet chain for
act $x$, but that capacity may vary with the act's overall
certainty-equivalent level. These are different restrictions on context dependence rather than a claim
that one class is generally more or less general than the other. Cumulative
utility can let the weight on $(z,A)$ vary across consequence levels within a
single act, while the present model requires one capacity to govern all steps
of an act's Choquet chain but can vary that capacity with the act's overall
certainty equivalent. Conversely, two acts can share a particular consequence
level while having different certainty equivalents and hence different local
capacities here. These restrictions point in different directions, so the paper
does not assert a general nesting relation; fixed-capacity Choquet expected
utility is a clear common benchmark.

Taken together, these comparisons place the model between two familiar
benchmarks. The implicit literature permits supporting evaluations to vary
across indifference levels, while the Choquet and cumulative literatures permit
nonadditive or consequence-sensitive state weighting. The distinctive
object here is a single event capacity at each
certainty-equivalent level, with its weights identified from binary event
comparisons. The paper therefore does not claim novelty for implicit
linearization, utility-dependent rank weighting, or the Choquet integral in
isolation.

\subsection{Ambiguity perception and ambiguity attitudes}
\label{subsec:related-attitudes}

The interpretation of a changing capacity also relates to work separating
ambiguity from attitudes toward it. \citet{GhirardatoMarinacci2002} provide a
comparative behavioral foundation for ambiguity and ambiguity aversion, while
\citet{GhirardatoMaccheroniMarinacci2004} derive an unambiguous preference
relation and use it to distinguish revealed ambiguity from ambiguity attitude.
\citet{KlibanoffMarinacciMukerji2005} make the separation explicit through
second-order beliefs and a distinct ambiguity-attitude transformation.
\citet{DentiPomatto2022} provide an Anscombe--Aumann axiomatic foundation for
smooth ambiguity preferences and show that the relevant parameters can be
uniquely recovered from preferences, making the separation between ambiguity
perception and ambiguity attitude operational. \citet{Siniscalchi2009} instead
evaluates an act through a baseline expected
utility plus adjustments reflecting its exposure to distinct sources of
ambiguity and the decision maker's attitudes toward that exposure.

These papers motivate the deliberately cautious interpretation adopted here.
The baseline theorem identifies $\nu_v$ as a local ambiguity \emph{weighting};
it does not by itself determine whether a change in $\nu_v$ reflects a change
in perceived ambiguity, a change in attitude toward a fixed ambiguity source,
or both. The convex-capacity results in Section~\ref{sec:ambiguity-analysis}
provide a stronger interpretation through the certainty-equivalent-indexed local core
$\mathcal P_v$, while further behavioral restrictions would be required for a
full perception--attitude separation.

A closely related comparative-statics literature studies how ambiguity
attitudes change with the decision maker's position. In particular,
\citet{CerreiaVioglioMaccheroniMarinacci2022} characterize wealth-dependent
absolute and relative ambiguity attitudes and distinguish wealth shifts from
utility shifts. Their exercise compares ambiguity attitudes across exogenously
shifted environments; the present representation instead endogenously indexes
the event-weighting object itself by the certainty equivalent and asks when
that object is a Choquet capacity. An earlier Anscombe--Aumann benchmark, \citet{GrantPolak2013}, axiomatizes
mean--dispersion preferences with constant absolute uncertainty aversion,
placing certainty translations directly at the center of ambiguity-attitude
invariance. In a recursive ambiguity framework,
\citet{MarinacciPrincipiStanca2026} show that translation invariance implies
ambiguity attitudes that are constant across welfare levels. Proposition~\ref{prop:translation-fixed}
provides a complementary static characterization by showing that, within the
implicit-Choquet class, the absence of certainty-translation effects is equivalent
to constancy of the entire capacity schedule.

The experimental literature reinforces the usefulness of allowing the local
weighting to move. \citet{BaillonBleichrodt2015} document a fourfold pattern of
ambiguity attitudes across gains, losses, and likelihood ranges, while
\citet{BouchouichaEtAl2017} find systematic stake effects. Such evidence is
broader than the present model's certainty-equivalent-level channel and may also involve
reference dependence or probability sensitivity. It is therefore best viewed
as complementary motivation rather than as evidence uniquely selecting the
present representation.

\subsection{Cognitive ambiguity attitudes}
\label{subsec:payro-takeoka-xia}

A particularly close recent model of endogenous nonadditive ambiguity
attitudes is \citet{PayroTakeokaXia2026}. Their Optimal Ambiguity
Attitude representation
lets the decision maker choose an evaluative capacity against a cognitive cost.
Schematically, letting $\mathcal V$ denote the feasible family of capacities,
\(U^{\mathrm{OAA}}(x) = \max_{\nu\in\mathcal V} \bigl\{C_\nu(x)-\kappa(\nu)\bigr\}\),
so the capacity selected for an act may vary with its payoff configuration.
The present model imposes a different discipline: all acts on the same
level-$v$ indifference surface share the same local capacity $\nu_v$, but the
capacity may change when the certainty-equivalent level changes.

Their key behavioral restrictions include Weak Certainty Independence and
Comonotonic Convexity. The latter permits a mixture of indifferent comonotonic
acts to be weakly worse when re-evaluating the mixture creates cognitive cost;
comonotonic mixture-betweenness in the present paper instead imposes local
neutrality for such mixtures. Weak Certainty Independence provides an even
sharper cross-level distinction. It is useful here as a literature-comparison
axiom rather than as one of the internal benchmark restrictions in
Section~\ref{sec:special-cases}.

Say that preferences satisfy \emph{Weak Certainty Independence} if, for all
$x,y\in X$, constants $c\mathbf1,d\mathbf1$, and $\alpha\in(0,1)$ for which
the mixtures are feasible,
\(\alpha x+(1-\alpha)c\mathbf1 \succeq \alpha y+(1-\alpha)c\mathbf1 \;\Longrightarrow\; \alpha x+(1-\alpha)d\mathbf1 \succeq \alpha y+(1-\alpha)d\mathbf1\).

Within the present representation, Weak Certainty Independence is especially
informative because it tests whether the certainty-equivalent-level index can affect the
capacity at all.

\begin{corollary}
\label{cor:wci-fixed}
Suppose the axioms of Theorem~\ref{thm:implicit-choquet} hold and, in addition,
preferences satisfy Weak Certainty Independence. Then there exists a single
normalized monotone capacity $\nu$, satisfying
$0<\nu(A)<1$ for every nontrivial event $A$, such that \( V(x)=C_\nu(x)\), so $\nu_v=\nu$ for every interior certainty-equivalent level $v$.
\end{corollary}

Thus Weak Certainty Independence shuts down the paper's distinctive
cross-level margin and returns the fixed-capacity Choquet benchmark.

This restriction is also empirically nontrivial. \citet{TrautmannWakker2018}
report violations of Weak Certainty Independence in a simple
Anscombe--Aumann experiment, alongside reference dependence and departures
from universal ambiguity aversion. More directly,
\citet{KonigKerstingKopsTrautmann2023} test certainty independence and Weak
Certainty Independence and find systematic violations of restrictions retained
by prominent ambiguity models. These findings do not constitute a direct test
of utility-level dependence, but they support relaxing cross-level invariance
conditions that, within the present representation, force $\nu_v$ to be
constant.

The same CEU implication appears in \citet{PayroTakeokaXia2026}, under their Basic Axioms, Weak Certainty
Independence, and Comonotonic Neutrality. For the present model, however, the
corollary has a short direct proof. Weak Certainty Independence preserves an
indifference when the common certainty component of the two mixtures is
changed. Applied to $x\sim V(x)\mathbf1$, this implies certainty translation
invariance of $V$ for every feasible common translation. Proposition
\ref{prop:translation-fixed} then forces the capacity schedule to be constant.
The appendix gives the argument explicitly. This also clarifies the economic
difference: cognitive OAA allows the evaluative capacity to vary with the act
while preserving the relevant certainty-independence restriction, whereas the present model allows the act's certainty-equivalent level to index
the local capacity. This level should not be identified with an exogenous
wealth state, although common translations provide a behavioral comparative
static connecting the two ideas.

When the capacities in the cognitive model are convex, each Choquet term can
also be expressed as a minimum over its core. Their criterion then has the
schematic form
\(U^{\mathrm{OAA}}(x) = \max_{\mathcal P} \left\{ \min_{p\in\mathcal P}p\cdot x -\kappa(\mathcal P) \right\}\),
whereas the convex subclass of the present model gives
\(
    V(x)
    =
    \min_{p\in\mathcal P_{V(x)}}p\cdot x.
\)
The first optimizes over ambiguity evaluations for a given act; the second
reveals a schedule of locally relevant ambiguity sets across indifference
levels.

\subsection{Multiple priors and general uncertainty aversion}
\label{subsec:related-multiple-priors}

The convex-capacity results also connect the paper to the broader ambiguity-aversion literature. \citet{GilboaSchmeidler1989} evaluate every act against a
fixed set of priors, while \citet{MaccheroniMarinacciRustichini2006} replace
the fixed set by a global penalty over priors. \citet{ChateauneufFaro2009}
use confidence functions over probabilities to represent ambiguity-sensitive
preferences, and \citet{CerreiaVioglioEtAl2011} characterize the broad class of
complete, transitive, monotone, and convex uncertainty-averse preferences.
Relative to these models, the distinctive margin is that the locally relevant
ambiguity set may vary with the certainty-equivalent level. In the convex
subclass, the cores $\mathcal P_v=\operatorname{core}(\nu_v)$ therefore make
it possible to study expansion, contraction, or other cross-level changes in
the relevant set of probability models; Section~\ref{sec:ambiguity-analysis}
develops these comparative implications.

\section{Conclusion}
\label{sec:conclusion}

This paper characterizes a disciplined form of certainty-equivalent-dependent
ambiguity evaluation. After the standard Anscombe--Aumann calibration of
consequence utility, the behavioral axioms yield a unique continuous family of
capacities $\{\nu_v\}_{v\in(0,1)}$. Every nonendpoint act is evaluated using
the capacity associated with its own certainty-equivalent level, while all
acts on the same interior indifference surface share the same event weighting.
The representation therefore lies between a globally fixed capacity and an
arbitrary act-dependent weighting rule.

Three implications give the representation its economic content. First, local
uncertainty aversion is equivalent to convexity of $\nu_v$ and yields a
certainty-equivalent-indexed multiple-priors interpretation. Second, binary
event threshold comparisons point-identify $\nu_v(A)$ at each elicited level,
so choice can reveal how the evaluation of the same event changes across
certainty-equivalent levels; in the convex subclass, uniform movements of the
capacity correspond to expansion or contraction of the locally relevant core.
Third, certainty translation invariance holds if and only if the capacity
schedule is constant. Fixed-capacity Choquet expected utility is therefore
exactly the benchmark in which feasible common certainty shifts leave the
uncertainty-evaluation object unchanged. Throughout, the capacity schedule is
a reduced-form ambiguity weighting: its variation may reflect changing
ambiguity perception, changing ambiguity attitude, or both.

Several extensions are natural, including richer comparative orders on the
capacity schedule, dynamic updating, behavioral separation of ambiguity
perception from ambiguity attitude, richer state spaces, and empirical designs
for tracing event-wise profiles. The central message is that systematic
cross-level variation in ambiguity weighting need not be left unrestricted:
it can be organized by certainty-equivalent indifference surfaces, recovered
from choice, and linked to sharp behavioral restrictions on familiar
benchmarks.
\appendix

\section{Auxiliary results for the representation theorem}
\label{app:auxiliary}

The results in this appendix provide the local-linear, gluing, continuity,
and reverse-implication steps used in Theorem \ref{thm:implicit-choquet}.

\subsection{Certainty equivalents and local implicit linearity}
\label{sec:local}

The first step fixes the utility scale. Continuity and monotonicity place every
act between its worst and best state utility, while strict ordering of constant
acts makes the resulting certainty equivalent unique.

\begin{lemma}
\label{lem:certainty-equivalent}
Under Axioms \ref{ax:weak-order-continuity}--\ref{ax:strict-constants},
for every \(x\in X\) there exists a unique \(V(x)\in I\) such that \(x\sim V(x)\1\). Moreover, \(V:X\to I\) is continuous and represents \(\succeq\):
\begin{equation}
    x\succeq y
    \quad\text{if and only if}\quad
    V(x)\geq V(y).
    \label{eq:ce-representation}
\end{equation}
If Axiom \ref{ax:boundary-interiority} also holds, then
\(V(x)\in(0,1)\) for every
\(x\in X\setminus\{\mathbf0,\mathbf1\}\).
\end{lemma}

\begin{proof}
Statewise monotonicity gives
\(
    (\min_s x_s)\1
    \preceq x\preceq
    (\max_s x_s)\1.
\)
Continuity and connectedness of the constant-act segment therefore yield a
constant \(c\) between \(\min_s x_s\) and \(\max_s x_s\) such that
\(x\sim c\1\). Axiom \ref{ax:strict-constants} makes \(c\) unique;
define it as \(V(x)\).
Transitivity and the strict ordering of constants then give
\eqref{eq:ce-representation}.

For continuity, let \(x^k\to x\). Compactness of \(I\) implies that any
subsequence of \(V(x^k)\) has a convergent subsubsequence, say
\(V(x^{k_j})\to c\). Write \(r:=V(x)\). We show directly from the closed
contour sets that \(c=r\).

Suppose first that \(c<r\), and choose \(d\) with \(c<d<r\). Since
\(
    x\sim r\mathbf1\succ d\mathbf1,
\)
we have \(x\succ d\mathbf1\). The lower contour set
\(
    \{y\in X:d\mathbf1\succeq y\}
\)
is closed by Axiom \ref{ax:weak-order-continuity}; by completeness, its
complement is exactly the strict upper set
\(
    \{y\in X:y\succ d\mathbf1\}.
\)
Hence this strict upper set is open. Since it contains \(x\), for all
sufficiently large \(j\),
\(
    x^{k_j}\succ d\mathbf1.
\)
The certainty-equivalent representation then gives
\(
    V(x^{k_j})>d,
\)
contradicting \(V(x^{k_j})\to c<d\).

If instead \(c>r\), choose \(d\) with \(r<d<c\). Then
\(
    d\mathbf1\succ r\mathbf1\sim x.
\)
The upper contour set
\(
    \{y\in X:y\succeq d\mathbf1\}
\)
is closed, so by completeness its complement
\(
    \{y\in X:d\mathbf1\succ y\}
\)
is open and contains \(x\). Therefore, for all sufficiently large \(j\),
\(
    d\mathbf1\succ x^{k_j},
\)
which implies \(V(x^{k_j})<d\), contradicting
\(V(x^{k_j})\to c>d\). Thus every convergent subsubsequence has limit
\(c=r=V(x)\), and consequently
\(
    V(x^k)\to V(x).
\)

If $x\in(0,1)^S$, statewise monotonicity already gives
$V(x)\in(0,1)$. If $x$ is a nonconstant boundary act, Axiom
\ref{ax:boundary-interiority} gives an interior constant indifferent to $x$,
and uniqueness of the certainty equivalent identifies its value with $V(x)$.
\end{proof}

This lemma fixes the cardinal level index used throughout the construction:
all subsequent local representations can be indexed by the unique certainty
equivalent \(V(x)\).

We now apply mixture-betweenness locally. For each rank cone, the restriction
of the preference satisfies ordinary mixture-betweenness, so \citet{Payro2025}'s general
mixture-space theorem provides an implicit mixture-linear functional at every
certainty-equivalent level. The next lemma specializes that result to the present
finite-dimensional domain.

\begin{lemma}
\label{lem:local-linear}
Suppose Axioms \ref{ax:weak-order-continuity}--
\ref{ax:boundary-interiority} hold. Fix a
permutation \(\sigma\). There exists a function
\(
    q^\sigma:(0,1)\to\R^S
\)
continuous in \(v\), with 
\begin{equation}
    \sum_{s\in S}q_s^\sigma(v)=1,
    \label{eq:local-sum-one}
\end{equation}
such that, for every \(x\in K_\sigma\) and every \(v\in(0,1)\),
\begin{equation}
    V(x)=v
    \quad\text{if and only if}\quad
    q^\sigma(v)\cdot x=v.
    \label{eq:local-root}
\end{equation}
\end{lemma}

\begin{proof}
Every pair of acts in \(K_\sigma\) is comonotonic. Hence the restriction of
\(\succeq\) to \(K_\sigma\) satisfies Weak Order, Continuity, and ordinary
Mixture-Betweenness. The set \(K_\sigma\) is nonempty, compact, convex, and
closed under mixtures, so it is a mixture space on which the restricted
preference is nontrivial. The constant acts \(\mathbf0\) and \(\mathbf1\) are
respectively weak worst and weak best by statewise monotonicity. Moreover,
Lemma \ref{lem:certainty-equivalent} together with Axiom
\ref{ax:boundary-interiority} gives \(V(x)\in(0,1)\) for every
\(x\in X\setminus\{\mathbf0,\mathbf1\}\). Hence, for every such \(x\),
\(x\sim V(x)\mathbf1\), and Axiom \ref{ax:strict-constants} implies
\(\mathbf1\succ x\succ\mathbf0\). Thus \(\mathbf0\) and \(\mathbf1\)
are respectively strict worst and strict best acts on \(K_\sigma\), so the
extremal hypotheses and normalization required for the local application below
are satisfied.

Apply Theorem 3.1 of \citet{Payro2025} to this restricted preference, using
\(\mathbf0\) and \(\mathbf1\) as the normalizing worst and best acts. It yields
\(
    \Phi_\sigma:K_\sigma\times[0,1]\to\R
\)
such that \(\Phi_\sigma(\cdot,v)\) is mixture-linear for every \(v\),
\(\Phi_\sigma(\mathbf0,v)=0\), \(\Phi_\sigma(\mathbf1,v)=1\), and \(V(x)\)
is the unique solution of
\(
    v=\Phi_\sigma(x,v).
\)
Since \(K_\sigma\) is full-dimensional, mixture-linearity is ordinary affine
linearity. The normalization at \(\mathbf0\) removes the intercept, so there
is a vector \(q^\sigma(v)\in\R^S\) satisfying
\(
    \Phi_\sigma(x,v)=q^\sigma(v)\cdot x.
\)
The normalization at \(\mathbf1\) gives
\(
    q^\sigma(v)\cdot\mathbf1=1,
\)
which is \eqref{eq:local-sum-one}, while \citet{Payro2025}'s unique-root property gives
\eqref{eq:local-root}.

It remains only to verify continuity of the coefficients. For
\(i=1,\ldots,n\), let
\(b_i^\sigma:=\mathbf1_{A_i^\sigma}, \; b_{n+1}^\sigma:=\mathbf0\).
Each \(b_i^\sigma\) belongs to \(K_\sigma\), and
\(
    \Phi_\sigma(b_i^\sigma,v)
    =
    \sum_{j=i}^{n}q_{\sigma(j)}^\sigma(v).
\)
Therefore
\(q_{\sigma(i)}^\sigma(v) = \Phi_\sigma(b_i^\sigma,v) - \Phi_\sigma(b_{i+1}^\sigma,v)\).
\citet{Payro2025}'s continuity of \(\Phi_\sigma\) in its second argument on \((0,1)\)
therefore implies continuity of every coordinate of \(q^\sigma\).
\end{proof}

Thus each rank cone admits a level-indexed supporting probability vector.
The remaining task is to make these cone-specific objects compatible across
rankings.

The unique-root representation in Lemma \ref{lem:local-linear} identifies the
indifference hyperplane. To use statewise monotonicity, we also need to know
which side of that hyperplane is preferred. Betweenness determines the
orientation.

\begin{lemma}
\label{lem:local-threshold}
Under Axioms \ref{ax:weak-order-continuity}--
\ref{ax:boundary-interiority}, with \(q^\sigma\) as in Lemma
\ref{lem:local-linear}, for every \(x\in K_\sigma\) and \(v\in(0,1)\),
\(x\succeq v\1 \;\text{if and only if}\; q^\sigma(v)\cdot x\geq v\).

\end{lemma}

\begin{proof}
The weak statement has two cases: strict and equality. The equality statement follows immediately from
\eqref{eq:local-root} and the definition of the certainty equivalent.
Consider the strict upper contour. Suppose, toward a contradiction, that
\(x\succ v\1\) but
\(
    q^\sigma(v)\cdot x<v.
\)
Choose a constant \(c\in(v,1)\). Then \(c\1\succ v\1\), and
\(
    q^\sigma(v)\cdot(c\1)=c>v.
\)
Hence there exists \(\lambda\in(0,1)\) such that
\(
    z
    :=
    \lambda x+(1-\lambda)c\1
\)
satisfies \(q^\sigma(v)\cdot z=v\). The constant act is comonotonic with
\(x\), so Axiom \ref{ax:comonotonic-betweenness} implies that \(z\) is
strictly preferred to \(v\1\): if \(x\succ c\1\), then
\(x\succ z\succ c\1\succ v\1\); if \(c\1\succ x\), then
\(c\1\succ z\succ x\succ v\1\); and if \(x\sim c\1\), then
\(z\sim x\succ v\1\). This contradicts the equality statement, which gives \(z\sim v\1\).
Thus
\(
    x\succ v\1\Rightarrow q^\sigma(v)\cdot x>v.
\)
The reverse implication follows by applying the same argument to the strict
lower contour, using a constant \(c\in(0,v)\). The weak statement follows from
completeness together with the strict and equality cases.
\end{proof}

The lemma upgrades the local root equation from a description of an
indifference hyperplane to a full threshold representation, which is needed
for both monotonicity and cross-cone gluing.

Statewise monotonicity now turns each local hyperplane normal into a genuine
probability vector.

\begin{lemma}
\label{lem:local-probability}
Under Axioms \ref{ax:weak-order-continuity}--
\ref{ax:boundary-interiority}, \(q^\sigma(v)\in\Delta(S)
    \;
    \text{for every }\sigma\text{ and }v\in(0,1)\),
where \(q^\sigma\) is the local representation from Lemma
\ref{lem:local-linear}.

\end{lemma}

\begin{proof}
By Lemma \ref{lem:local-linear}, the coordinates sum to one. It remains to
show nonnegativity.

Fix \(v\in(0,1)\). Choose numbers \(\delta>0\) and \(\varepsilon>0\), small
enough that all coordinates below lie in \((0,1)\), with
\(n\varepsilon<\delta\) and \(\delta+n\varepsilon<\min\{v,1-v\}\). Define two vectors in the strict interior of
\(K_\sigma\) by
\(p_{\sigma(i)}:=v+\delta+i\varepsilon, \; n_{\sigma(i)}:=v-\delta+i\varepsilon, \; i=1,\ldots,n\).
Then every coordinate of \(p\) exceeds \(v\), every coordinate of \(n\) is
below \(v\), and the whole line segment joining \(p\) to \(n\) has the same
strict ranking. More explicitly,
\(p\succeq (\min_s p_s)\1\succ v\1 \;\text{and}\; v\1\succ (\max_s n_s)\1\succeq n\),
where the strict comparisons use Axiom~\ref{ax:strict-constants}. Hence
\(
    p\succ v\1\succ n.
\)
By continuity, the segment contains some strict-ranking point
\(e\in K_\sigma\) with \(e\sim v\1\).

Let \(\delta^s\) denote the \(s\)th unit vector. Suppose
\(q_s^\sigma(v)<0\) for some state \(s\). Since \(e\) lies in the strict
interior of the rank cone and all its coordinates lie in \((0,1)\), for
sufficiently small \(\eta>0\), \(e+\eta\delta^s\in K_\sigma\subseteq X\).
Moreover, \(e+\eta\delta^s\geq e\) statewise, so statewise monotonicity
gives \(e+\eta\delta^s\succeq e\). Since \(e\sim v\1\), transitivity
therefore yields \(e+\eta\delta^s\succeq v\1\). But Lemma
\ref{lem:local-threshold} gives
\(
    q^\sigma(v)\cdot(e+\eta\delta^s)
    =v+\eta q_s^\sigma(v)
    <v,
\)
which implies \(v\1\succ e+\eta\delta^s\), a contradiction. Thus every
coordinate is nonnegative, and so \( q^\sigma(v)\in\Delta(S)
    \;
    \text{for every }\sigma\text{ and }v\in(0,1)\) follows.
\end{proof}

Hence the local hyperplane coefficients are genuine state weights rather
than unrestricted affine coefficients.

\subsection{Gluing rank cones into a certainty-equivalent-indexed capacity}
\label{sec:gluing}

The remaining issue is global consistency. Different rank cones may have
different local probability vectors \(q^\sigma(v)\). The relevant consistency
condition is not equality of these vectors. Instead, whenever an event \(A\)
is ranked above its complement, the aggregate local weight assigned to \(A\)
must be independent of how states are ordered within \(A\) and within
\(A^c\). Binary event acts make this implication transparent.

For \(A\subseteq S\), let \(\Sigma(A)\) denote the permutations under which
every state in \(A\) is ranked weakly above every state in \(A^c\). For
\(\sigma\in\Sigma(A)\), define \( r_\sigma(A,v)
    :=
    \sum_{s\in A}q_s^\sigma(v)\).

\begin{lemma}
\label{lem:binary-gluing}
For every \(v\in(0,1)\), every event \(A\subseteq S\), and every
\(\sigma,\tau\in\Sigma(A)\), 
\begin{equation}
    r_\sigma(A,v)=r_\tau(A,v).
    \label{eq:gluing-equality}
\end{equation}
\end{lemma}

\begin{proof}[Proof of Lemma \ref{lem:binary-gluing}]
Fix \(v\), \(A\), and two admissible rankings \(\sigma,\tau\). For positive
\(a,b\), small enough that \(v+a\leq1\) and \(v-b\geq0\), consider the binary
act 
\begin{equation}
    x^{A}_{a,b}
    :=
    (v+a)\mathbf1_A
    +(v-b)\mathbf1_{A^c}.
    \label{eq:binary-event-act}
\end{equation}
For later one-sided identifications, we use the same notation when one of
\(a,b\) is zero, provided \(a,b\geq0\), \(a+b>0\), and the resulting act is
feasible. The gluing argument here continues to use \(a,b>0\).
This act belongs to both \(K_\sigma\) and \(K_\tau\). Lemma
\ref{lem:local-threshold} implies, under ranking \(\sigma\),
\begin{equation}
    x^{A}_{a,b}\succeq v\1
    \quad\text{if and only if}\quad
    a\,r_\sigma(A,v)
    -b\,[1-r_\sigma(A,v)]
    \geq0.
    \label{eq:binary-sign-sigma}
\end{equation}
The same primitive comparison, evaluated using \(\tau\), is equivalent to
\begin{equation}
    a\,r_\tau(A,v)
    -b\,[1-r_\tau(A,v)]
    \geq0.
    \label{eq:binary-sign-tau}
\end{equation}
If the two aggregate weights differed, choose the ratio \(a/b\) strictly
between the two corresponding cutoff ratios
\(
    (1-r)/r
\)
(with the usual endpoint interpretation when \(r=0\) or \(r=1\)). Scaling
\(a\) and \(b\) by a common positive factor keeps this ratio unchanged and
ensures that \eqref{eq:binary-event-act} remains in \(X\). Then
\eqref{eq:binary-sign-sigma} and \eqref{eq:binary-sign-tau} would assign
opposite rankings to the same pair of acts, a contradiction. Hence
\eqref{eq:gluing-equality} holds.
\end{proof}

This is the key compatibility step: the aggregate weight assigned to an
event is identified independently of the ranking chosen within the event and
its complement.

Lemma \ref{lem:binary-gluing} makes it possible to define an event weight
without selecting a particular rank cone.

\begin{proposition}
\label{prop:capacity-construction}
For every \(v\in(0,1)\), define
\begin{equation}
    \nu_v(A)
    :=
    \sum_{s\in A}q_s^\sigma(v),
    \qquad
    \sigma\in\Sigma(A).
    \label{eq:capacity-definition}
\end{equation}
Then:
\begin{enumerate}[label=(\roman*),leftmargin=2.5em]
    \item \(\nu_v\) is well defined and is a normalized monotone capacity;
    \item for every rank cone \(K_\sigma\) and every \(x\in K_\sigma\),
    \(C_{\nu_v}(x) =q^\sigma(v)\cdot x\);
    \item for every \(x\in X\),
    \begin{equation}
        x\succeq v\1
        \quad\text{if and only if}\quad
        C_{\nu_v}(x)\geq v.
        \label{eq:global-threshold}
    \end{equation}
\end{enumerate}
\end{proposition}

\begin{proof}
Well-definedness follows from Lemma \ref{lem:binary-gluing}. The normalization
\(
    \nu_v(\varnothing)=0
\)
and
\(
    \nu_v(S)=1
\)
follow from \eqref{eq:capacity-definition} and
\(q^\sigma(v)\in\Delta(S)\).

To show monotonicity, take \(A\subseteq B\). Choose a ranking whose blocks,
from high to low, are \(A\), \(B\setminus A\), and \(B^c\). Both \(A\) and
\(B\) are upper events under that ranking. Therefore
\(
    \nu_v(B)-\nu_v(A)
    =
    \sum_{s\in B\setminus A}q_s^\sigma(v)
    \geq0.
\)
Hence \(\nu_v\) is a capacity.

For the Choquet identification, fix \(x\in K_\sigma\) and use the upper events
\(A_i^\sigma\) from \eqref{eq:upper-events}. Since each \(A_i^\sigma\) is an
upper event under \(\sigma\),
\(\nu_v(A_i^\sigma) = \sum_{j=i}^{n}q_{\sigma(j)}^\sigma(v)\).
Consequently,
\(q_{\sigma(i)}^\sigma(v) = \nu_v(A_i^\sigma)-\nu_v(A_{i+1}^\sigma)\).
Substitution into the finite Choquet formula
\eqref{eq:choquet-integral} gives
\(
    C_{\nu_v}(x)=q^\sigma(v)\cdot x.
\)
Combining this equality with Lemma \ref{lem:local-threshold} yields
\eqref{eq:global-threshold}. Since the rank cones cover \(X\), the result is
global.
\end{proof}

The proposition completes the forward geometric construction by turning the
cone-specific probability vectors into one Choquet capacity at each utility
level.

The representation also identifies the capacity directly from preferences.
For an event \(A\), the number \(\nu_v(A)\) determines the cutoff gain--loss
tradeoff around the certainty-equivalent level \(v\) at which the binary event
comparison changes sign. When this cutoff is interior, it is attained by an
indifferent binary bet.

\begin{proposition}
\label{prop:revealed-capacity}
Fix \(v\in(0,1)\) and \(A\subseteq S\). For any \(a,b\geq0\) with
\(a+b>0\) such that the binary act in \eqref{eq:binary-event-act} is feasible,
\begin{equation}
    x^{A}_{a,b}\sim v\1
    \quad\text{if and only if}\quad
    a\nu_v(A)=b[1-\nu_v(A)].
    \label{eq:revealed-capacity-condition}
\end{equation}
In particular, whenever an interior indifference with \(a,b>0\) exists,
\begin{equation}
    \nu_v(A)=\frac{b}{a+b}.
    \label{eq:revealed-capacity-ratio}
\end{equation}
At the endpoints, \(x^A_{a,0}\sim v\mathbf1\) for \(a>0\) if and only if
\(\nu_v(A)=0\), while \(x^A_{0,b}\sim v\mathbf1\) for \(b>0\) if and only if
\(\nu_v(A)=1\). The capacity \(\nu_v\) is unique among normalized monotone
capacities satisfying \eqref{eq:global-threshold} at level \(v\).
\end{proposition}

\begin{proof}
For a binary act with \(A\) as its high event, the Choquet integral is
\(
    C_{\nu_v}(x^A_{a,b})
    =
    v-b+(a+b)\nu_v(A).
\)
By Proposition \ref{prop:capacity-construction} and the local root
identification used in its construction,
\(x^A_{a,b}\sim v\mathbf1 \;\Longleftrightarrow\; C_{\nu_v}(x^A_{a,b})=v\).
Substituting the binary-act formula gives
\(-b+(a+b)\nu_v(A)=0\),
which is \eqref{eq:revealed-capacity-condition}; when \(a,b>0\), rearranging
gives \eqref{eq:revealed-capacity-ratio}. The same equality with \(b=0\) or
\(a=0\) yields the stated one-sided endpoint identifications.

For uniqueness, suppose \(\mu_v\) is another normalized monotone capacity
that gives the same threshold comparisons at level \(v\). If
\(\mu_v(A)\neq\nu_v(A)\) for some event \(A\), choose the ratio \(a/b\) between
the corresponding binary-act cutoff ratios, with the usual one-sided
interpretation if one of the event weights is \(0\) or \(1\). After scaling
\(a,b\) to maintain feasibility, the two capacities rank the same binary act
on opposite sides of \(v\mathbf1\), contradicting the common threshold
representation. Thus \(\mu_v(A)=\nu_v(A)\) for every event.
\end{proof}

Besides proving uniqueness, binary event comparisons give each local capacity
a direct revealed-preference identification; when the cutoff is interior,
\eqref{eq:revealed-capacity-ratio} gives the corresponding indifference-ratio
elicitation.

The continuity of preferences has an additional implication: the local
capacity cannot jump as the certainty-equivalent level changes.

\begin{proposition}
\label{prop:capacity-continuity}
The map \( v\longmapsto\nu_v\) is continuous on \((0,1)\) under the sup norm \( \|\nu-\mu\|_\infty
    :=
    \max_{A\subseteq S}|\nu(A)-\mu(A)|\).

\end{proposition}

\begin{proof}
Fix an event \(A\), and choose once and for all a permutation
\(\sigma_A\in\Sigma(A)\). By definition,
\(
    \nu_v(A)
    =
    \sum_{s\in A}q_s^{\sigma_A}(v).
\)
Lemma \ref{lem:local-linear} establishes continuity of every coordinate of
\(q^{\sigma_A}(v)\), hence \(v\mapsto\nu_v(A)\) is continuous. Since \(S\) is
finite, \(2^S\) is finite. Coordinatewise continuity of the finitely many
numbers \(\{\nu_v(A):A\subseteq S\}\) is therefore equivalent to continuity
under the sup-norm.
\end{proof}

Thus the capacity schedule inherits continuity from the underlying local
representation rather than requiring an unrelated smoothness assumption.

\subsection{Auxiliary results for the reverse implication}

For the converse direction, uniqueness of the implicit root is not enough by
itself: we also need to know which side of the root corresponds to a better or
worse act. Because admissibility applies to every act other than the endpoint
constants, the orientation can be proved without relying on interior state bounds.

\begin{lemma}
\label{lem:single-crossing}
Let $\{\nu_v\}$ be admissible, and for
$x\in X\setminus\{\mathbf0,\mathbf1\}$ let $V(x)\in(0,1)$ be the unique
solution of $C_{\nu_v}(x)=v$. Then
\begin{equation}
    C_{\nu_v}(x)-v
    \begin{cases}
        >0,& v<V(x),\\
        =0,& v=V(x),\\
        <0,& v>V(x).
    \end{cases}
    \label{eq:single-crossing-sign}
\end{equation}
\end{lemma}

\begin{proof}
Fix $x\in X\setminus\{\mathbf0,\mathbf1\}$ and write $r:=V(x)$. Let
$h_x(v):=C_{\nu_v}(x)-v$. By admissibility, $h_x(r)=0$ and $r$ is the
unique root in $(0,1)$.

Take $v<r$. If $h_x(v)=0$, uniqueness is contradicted. Suppose instead that
$h_x(v)<0$, so $C_{\nu_v}(x)<v<r$. Define
\(\beta := \frac{r-v}{r-C_{\nu_v}(x)} \in(0,1), \; a:=(1-\beta)r\),
and set
\(
    y:=a\mathbf1+\beta x.
\)
Since $0<r<1$ and $0<\beta<1$, every coordinate of $y$ lies strictly
between $0$ and $1$, so $y\in(0,1)^S$. Translation invariance and positive
homogeneity of the Choquet integral give
\(
    C_{\nu_r}(y)=a+\beta r=r
\)
and, by the definition of $\beta$,
\(
    C_{\nu_v}(y)
    =a+\beta C_{\nu_v}(x)
    =v.
\)
Thus the act $y$ has two distinct roots, $v$ and $r$, contrary to
admissibility. Hence $h_x(v)>0$ for every $v<r$.

The argument for $v>r$ is symmetric. If $h_x(v)>0$, define
\(\beta := \frac{v-r}{C_{\nu_v}(x)-r} \in(0,1), \; a:=(1-\beta)r\).
Then the same interior act $y=a\mathbf1+\beta x$ satisfies both
$C_{\nu_r}(y)=r$ and $C_{\nu_v}(y)=v$, again contradicting uniqueness.
Therefore $h_x(v)<0$ for every $v>r$, proving
\eqref{eq:single-crossing-sign}.
\end{proof}

This orientation converts the implicit equation into threshold comparisons
for every act other than the endpoint constants, including nonconstant boundary
acts, and is the bridge from an
admissible family back to the full-domain behavioral axioms.

With the orientation established, the remaining converse task is to verify
that the root-defined preference satisfies exactly the maintained behavioral
requirements on $X$.

\begin{lemma}
\label{lem:admissible-generates-axioms}
Let $\{\nu_v:v\in(0,1)\}$ be an admissible implicit-capacity family. Define
$V:X\to[0,1]$ by
\(
    V(\mathbf0):=0,
    \; \text{and}\;
    V(\mathbf1):=1,
\)
and, for $x\in X\setminus\{\mathbf0,\mathbf1\}$, by the unique solution of
\(
    C_{\nu_{V(x)}}(x)=V(x).
\)
Define $x\succeq y$ if and only if $V(x)\geq V(y)$. Then $\succeq$ satisfies
Axioms \ref{ax:weak-order-continuity}--\ref{ax:boundary-interiority} on $X$.
\end{lemma}

\begin{proof}
Completeness and transitivity follow directly from the real-valued
representation. For every $c\in(0,1)$, normalized capacities satisfy
$C_{\nu_v}(c\mathbf1)=c$, so the unique root for $c\mathbf1$ is $v=c$.
Together with the definitions at $\mathbf0$ and $\mathbf1$, this gives
$V(c\mathbf1)=c$ for every $c\in[0,1]$ and therefore strict ordering of
constant acts.

We next establish continuity of $V$ on the full act domain. The finite Choquet
integral is jointly continuous in $(x,\nu)$. In particular, \(|C_\nu(x)-C_\nu(y)|
    \leq
    \|x-y\|_\infty\),
and, for fixed $x$, \(|C_\nu(x)-C_\mu(x)|
    \leq
    \operatorname{osc}(x)
    \|\nu-\mu\|_\infty\), where $\operatorname{osc}(x)=\max_s x_s-\min_s x_s$. Thus continuity of
$v\mapsto\nu_v$ implies joint continuity of
$(x,v)\mapsto C_{\nu_v}(x)-v$ on $X\times(0,1)$.

Let $x^k\to x$ with $x\notin\{\mathbf0,\mathbf1\}$. Choose
$0<a<V(x)<b<1$. By Lemma \ref{lem:single-crossing},
\(
    C_{\nu_a}(x)-a>0,
    \;
    C_{\nu_b}(x)-b<0.
\)
For all sufficiently large $k$, the same strict inequalities hold with
$x^k$ in place of $x$. Lemma \ref{lem:single-crossing} then implies
$V(x^k)\in(a,b)$. Any subsequence therefore has a further subsequence with
$V(x^{k_j})\to r\in[a,b]$. Passing to the limit in the root equation gives
$C_{\nu_r}(x)=r$, so uniqueness yields $r=V(x)$. Hence
$V(x^k)\to V(x)$.

At the endpoints, the root equation and normalization imply
\(
    \min_s x_s\leq V(x)\leq\max_s x_s
\)
for every nonconstant $x$. Therefore, if $x^k\to\mathbf0$,
$0\leq V(x^k)\leq\max_sx_s^k\to0$; similarly, if
$x^k\to\mathbf1$, then $\min_sx_s^k\leq V(x^k)\leq1$ and hence
$V(x^k)\to1$. Thus $V$ is continuous on all of $X$.

For statewise monotonicity, suppose $x\geq y$. If $y=\mathbf0$, then
$V(y)=0\leq V(x)$; if $x=\mathbf1$, then
$V(x)=1\geq V(y)$. Otherwise both $x$ and $y$ are nonendpoint acts. Put
$v:=V(y)\in(0,1)$. Monotonicity of the Choquet integral gives
\(
    C_{\nu_v}(x)
    \geq
    C_{\nu_v}(y)
    =v.
\)
Lemma \ref{lem:single-crossing} therefore implies
\(
    V(x)\geq v=V(y).
\)

Finally, let $x,y$ be comonotonic and let
$z_\lambda:=\lambda x+(1-\lambda)y$ for $\lambda\in(0,1)$. Put
$a:=V(x)$ and $b:=V(y)$. If $1>a>b>0$, the usual comonotonic-affinity
argument applies: at level $a$,
\(
    C_{\nu_a}(z_\lambda)<a,
\)
while at level $b$,
\(
    C_{\nu_b}(z_\lambda)>b,
\)
so Lemma \ref{lem:single-crossing} gives
$b<V(z_\lambda)<a$.

If $a=1>b>0$, then $x=\mathbf1$. At level $b$,
$C_{\nu_b}(z_\lambda)>b$, while $z_\lambda\neq\mathbf1$, so
$b<V(z_\lambda)<1=a$. The case $1>a>b=0$ is symmetric because then
$y=\mathbf0$. If $a=1$ and $b=0$, then
$z_\lambda=\lambda\mathbf1$ and $V(z_\lambda)=\lambda$. Thus strict
betweenness holds in every case.

If $a=b=v\in(0,1)$, comonotonic affinity gives
\(
    C_{\nu_v}(z_\lambda)
    =
    \lambda v+(1-\lambda)v
    =v,
\)
so uniqueness gives $V(z_\lambda)=v$. If $a=b\in\{0,1\}$, then both
acts equal the corresponding endpoint constant and the conclusion is
immediate. Hence comonotonic mixture-betweenness holds on all of $X$.

Finally, if $x\in\partial X\setminus\{\mathbf0,\mathbf1\}$, admissibility
gives $V(x)\in(0,1)$, and $V(V(x)\mathbf1)=V(x)$. Thus
$x\sim V(x)\mathbf1$, establishing Axiom \ref{ax:boundary-interiority}.
\end{proof}

This closes the full-domain converse: admissibility of the interior capacity
schedule is sufficient for the complete axiomatic package on $X$.

\section{Proof of the main representation theorem}
\label{app:representation-proofs}

The auxiliary ingredients have been established in Appendix
\ref{app:auxiliary}. We now combine them to prove the full-domain
characterization.

\begin{proof}[Proof of Theorem \ref{thm:implicit-choquet}]
\emph{Forward implication.}\par
Lemma \ref{lem:certainty-equivalent} gives the continuous certainty-equivalent
representation $V:X\to[0,1]$. By Axiom \ref{ax:boundary-interiority}, every
$x\in X\setminus\{\mathbf0,\mathbf1\}$ has $V(x)\in(0,1)$. The local-linear,
orientation, probability-weight, and gluing results in Appendix
\ref{app:auxiliary} then construct a normalized monotone capacity $\nu_v$ at
every interior certainty-equivalent level. Proposition \ref{prop:capacity-construction}
gives
\(
    x\succeq v\mathbf1
    \;\text{if and only if}\;
    C_{\nu_v}(x)\geq v
    \; \text{for all}\; x\in X,
\)
which is \eqref{eq:main-threshold}. The proof of Lemma
\ref{lem:local-threshold} also establishes the corresponding strict and equality
relations; together with
$C_{\nu_v}(x)=q^\sigma(v)\cdot x$ on any rank cone containing $x$, these give the corresponding strict and equality threshold relations. Proposition
\ref{prop:capacity-continuity} gives continuity of
$v\mapsto\nu_v$.

Fix $x\in X\setminus\{\mathbf0,\mathbf1\}$. Choose a rank cone containing
$x$. Lemma \ref{lem:local-threshold} together with Proposition
\ref{prop:capacity-construction} gives
\(
    x\sim v\mathbf1
    \;\text{if and only if}\;
    C_{\nu_v}(x)=v.
\)
Since $V(x)\in(0,1)$, $v=V(x)$ is a root of \(C_{\nu_v}(x)=v\).
If another $v\in(0,1)$ were a root, the same
equivalence would imply $x\sim v\mathbf1$, contradicting uniqueness of the
certainty equivalent. Thus the family is admissible in the sense of Definition
\ref{def:admissible-family}.

It remains to establish uniqueness of the admissible representing family. Let
$\{\mu_v:v\in(0,1)\}$ be any other admissible representing family satisfying
part (ii) for the same preference, and let $W$ be its root-defined map.
For every $c\in(0,1)$, normalization gives
$C_{\mu_t}(c\mathbf1)=c$ for every $t\in(0,1)$, so the unique root for the
constant act $c\mathbf1$ is $t=c$. Together with the endpoint definitions,
$W(c\mathbf1)=c$ for every $c\in[0,1]$. Since $W$ represents the same
preference, $x\sim W(x)\mathbf1$ for every $x\in X$. Lemma
\ref{lem:certainty-equivalent} gives a unique certainty equivalent on this
normalization, and therefore $W(x)=V(x)$ for every $x\in X$.

Apply Lemma \ref{lem:single-crossing} to the admissible family $\{\mu_v\}$.
For every $v\in(0,1)$ and every nonendpoint act $x$,
\(C_{\mu_v}(x)\geq v \;\Longleftrightarrow\; W(x)\geq v \;\Longleftrightarrow\; V(x)\geq v \;\Longleftrightarrow\; x\succeq v\mathbf1\).
The endpoint acts satisfy the same equivalence by normalization. Thus, at each
level $v$, $\mu_v$ and $\nu_v$ generate the same threshold comparisons.
Proposition \ref{prop:revealed-capacity} then implies $\mu_v=\nu_v$ for every
$v\in(0,1)$. Hence the admissible representing family is unique.

The endpoint values are
$V(\mathbf0)=0$ and $V(\mathbf1)=1$, and
\eqref{eq:main-value-representation} is exactly the representation in Lemma
\ref{lem:certainty-equivalent}.

\emph{Reverse implication.}\par
Let an admissible family be given and define $V$ on $X$ as in part (ii) of the
theorem. Lemma \ref{lem:admissible-generates-axioms} shows that the preference
represented by $V$ satisfies Axioms
\ref{ax:weak-order-continuity}--\ref{ax:boundary-interiority} on the full
domain $X$. This proves the reverse implication and completes the proof of the
theorem.
\end{proof}

\paragraph{Threshold consequences.}
Take $v\in(0,1)$. If
$x\notin\{\mathbf0,\mathbf1\}$, Lemma \ref{lem:single-crossing} gives
\(V(x)>v \Longleftrightarrow C_{\nu_v}(x)>v, \; V(x)=v \Longleftrightarrow C_{\nu_v}(x)=v\),
and
\(V(x)<v \Longleftrightarrow C_{\nu_v}(x)<v\).
Since $V(v\mathbf1)=v$, these are exactly the corresponding strict,
equality, and reverse strict preference relations; the weak relation
\eqref{eq:main-threshold} follows as well. For
$x=\mathbf0$ or $x=\mathbf1$, the same relations follow directly from
normalization of the capacity. Hence all threshold statements hold for every
$x\in X$.

\section{Proofs of the remaining results}
\label{app:remaining-proofs}

\begin{proof}[Proof of Remark \ref{rem:admissibility-sufficient}]
We proceed in four steps.

\emph{Step 1: endpoint limits.}
Because \(S\) is finite, a capacity can be identified with a vector in
\([0,1]^{2^{|S|}}\), endowed with the sup norm. The set of normalized monotone
capacities is closed in this finite-dimensional space and hence complete.

Let \(v_n\downarrow0\). The Lipschitz condition implies
\(
    \|\nu_{v_n}-\nu_{v_m}\|_\infty
    \leq
    L|v_n-v_m|.
\)
Since \((v_n)\) is Cauchy, \((\nu_{v_n})\) is Cauchy under the sup norm and
therefore converges to some normalized monotone capacity, denoted by
\(\nu_0\). The limit is independent of the sequence approaching \(0\): if
\(w_n\downarrow0\) and \(\nu_{w_n}\to\widetilde\nu_0\), then
\(
    \|\nu_{v_n}-\nu_{w_n}\|_\infty
    \leq
    L|v_n-w_n|
    \longrightarrow0,
\)
so \(\nu_0=\widetilde\nu_0\). Hence
\(
    \nu_0
    :=
    \lim_{v\downarrow0}\nu_v
\)
is well defined. The same argument gives the unique limit
\(
    \nu_1
    :=
    \lim_{v\uparrow1}\nu_v.
\)
Indeed, taking limits in the Lipschitz inequality also gives
\(
    \|\nu_v-\nu_0\|_\infty\leq Lv
    \;\text{and}\;
    \|\nu_v-\nu_1\|_\infty\leq L(1-v).
\)

\emph{Step 2: endpoint signs.}
Fix
\(
    x\in X\setminus\{\mathbf0,\mathbf1\}.
\)
Since \(x\neq\mathbf0\), there exists a nonempty event \(A\subseteq S\) and
a number \(m>0\) such that
\(
    x\geq m\mathbf1_A
\)
statewise; for example, take \(A\) to be the set of states attaining
\(\max_s x_s\) and \(m=\max_s x_s\). By monotonicity of the Choquet integral,
\(
    C_{\nu_0}(x)
    \geq
    C_{\nu_0}(m\mathbf1_A)
    =
    m\nu_0(A)
    >0.
\)
Thus
\(
    C_{\nu_0}(x)>0.
\)

Likewise, since \(x\neq\mathbf1\), the event
\(
    B:=\{s\in S:x_s<1\}
\)
is nonempty. Because \(S\) is finite,
\(
    \delta
    :=
    \min_{s\in B}(1-x_s)
    >0.
\)
Hence
\(
    x
    \leq
    \mathbf1-\delta\mathbf1_B.
\)
By monotonicity,
\(
    C_{\nu_1}(x)
    \leq
    C_{\nu_1}(\mathbf1-\delta\mathbf1_B).
\)
The act on the right takes the value \(1-\delta\) on \(B\) and \(1\) on
\(B^c\), so
\(
    C_{\nu_1}(\mathbf1-\delta\mathbf1_B)
    =
    1-\delta+\delta\nu_1(B^c).
\)
Since \(B\neq\varnothing\), the event \(B^c\) is proper, and therefore
\(
    \nu_1(B^c)<1.
\)
Consequently,
\(
    C_{\nu_1}(x)
    \leq
    1-\delta+\delta\nu_1(B^c)
    <1.
\)
Thus
\(
    C_{\nu_1}(x)<1.
\)

\emph{Step 3: existence and uniqueness of the implicit root.}
Fix \(x\in X\), and define
\(
    h_x(v)
    :=
    C_{\nu_v}(x)-v,
    \; v\in(0,1).
\)
For any two capacities \(\nu,\mu\), the finite Choquet integral satisfies
\(
    |C_\nu(x)-C_\mu(x)|
    \leq
    \operatorname{osc}(x)\|\nu-\mu\|_\infty,
\)
where
\(
    \operatorname{osc}(x)
    :=
    \max_sx_s-\min_sx_s
    \leq1.
\)
Hence, for \(w>v\),
\begin{align*}
    h_x(w)-h_x(v)
    &=
    C_{\nu_w}(x)-C_{\nu_v}(x)-(w-v)\\
    &\leq
    |C_{\nu_w}(x)-C_{\nu_v}(x)|-(w-v)\\
    &\leq
    \operatorname{osc}(x)\|\nu_w-\nu_v\|_\infty-(w-v)\\
    &\leq
    L(w-v)-(w-v)\\
    &=
    -(1-L)(w-v)
    <0.
\end{align*}
Thus \(h_x\) is strictly decreasing.

Now suppose
\(
    x\notin\{\mathbf0,\mathbf1\}.
\)
By continuity of the Choquet integral with respect to the capacity,
\(
    \lim_{v\downarrow0}h_x(v)
    =
    C_{\nu_0}(x)
    >0,
\)
whereas
\(
    \lim_{v\uparrow1}h_x(v)
    =
    C_{\nu_1}(x)-1
    <0.
\)
Therefore \(h_x\) is positive for \(v\) sufficiently close to \(0\) and
negative for \(v\) sufficiently close to \(1\). Since \(h_x\) is continuous,
the intermediate value theorem gives at least one
\(v^*\in(0,1)\) satisfying
\(
    C_{\nu_{v^*}}(x)=v^*.
\)
Since \(h_x\) is strictly decreasing, this root is unique.

For an interior constant act \(x=c\mathbf1\), \(c\in(0,1)\), normalization
of every capacity gives
\(
    C_{\nu_v}(c\mathbf1)=c
    \;\text{for every }v,
\)
so
\(
    C_{\nu_v}(c\mathbf1)=v
    \;\text{if and only if}\;
    v=c.
\)
Hence its unique root is \(v=c\).

It follows that the schedule satisfies the unique-root requirement in
Definition \ref{def:admissible-family}. Together with the assumed Lipschitz
continuity and the fact that every \(\nu_v\) is a normalized monotone capacity,
the family is admissible.

\emph{Step 4: the explicit nonconstant family.}
Let
\(
    \nu_v
    =
    [1-\varepsilon g(v)]\nu^0
    +
    \varepsilon g(v)\nu^1,
    \; v\in(0,1),
\)
where \(\nu^0\neq\nu^1\),
\(g:[0,1]\to[0,1]\) is nonconstant and Lipschitz, and
\(
    \varepsilon\operatorname{Lip}(g)<1.
\)
Since
\(
    0\leq\varepsilon g(v)\leq1,
\)
each \(\nu_v\) is a convex combination of normalized monotone capacities and
is therefore itself a normalized monotone capacity.

Moreover, \( \|\nu_w-\nu_v\|_\infty=
    \varepsilon|g(w)-g(v)|
    \|\nu^1-\nu^0\|_\infty \leq
    \varepsilon\operatorname{Lip}(g)|w-v|\), because
\(
    \|\nu^1-\nu^0\|_\infty\leq1.
\)
Thus the schedule is Lipschitz with constant
\(
    L
    \leq
    \varepsilon\operatorname{Lip}(g)
    <1.
\)

Since \(g\) is Lipschitz, it is continuous on \([0,1]\), and the endpoint
limits are
\(
    \nu_0
    =
    [1-\varepsilon g(0)]\nu^0
    +
    \varepsilon g(0)\nu^1
\)
and
\(
    \nu_1
    =
    [1-\varepsilon g(1)]\nu^0
    +
    \varepsilon g(1)\nu^1.
\)
For every nonempty proper event \(A\),
\(
    0<\nu^k(A)<1,
    \; k=0,1.
\)
Hence every convex combination of \(\nu^0(A)\) and \(\nu^1(A)\) also lies
strictly between \(0\) and \(1\). In particular,
\(
    \nu_0(A)>0
    \;\text{for every nonempty }A,
\)
and
\(
    \nu_1(A)<1
    \;\text{for every proper }A.
\)
The preceding steps therefore imply that \(\{\nu_v\}_{v\in(0,1)}\) is
admissible.

The family is nonconstant because \(g\) is nonconstant,
\(\varepsilon>0\), and \(\nu^0\neq\nu^1\): there exists some event \(A\)
with \(\nu^0(A)\neq\nu^1(A)\), and hence \(v\mapsto\nu_v(A)\) is nonconstant. Finally, if \(\nu^0\) and \(\nu^1\) are convex capacities, then for every
\(v\), \(\nu_v\) is a convex combination of convex capacities. Since
supermodularity is preserved under convex combinations, every \(\nu_v\) is
convex.
\end{proof}

\begin{proof}[Proof of Proposition \ref{prop:global-betweenness}]
If mixture-betweenness holds globally, the maintained boundary-interiority and
strict-constant axioms make \(\mathbf0\) and \(\mathbf1\) respectively strict
worst and strict best acts, by Lemma \ref{lem:certainty-equivalent}. Apply
\citet{Payro2025}'s Theorem 3.1 directly to \(X\), rather than cone by cone.
Mixture-linearity and the normalization at
\(\mathbf0\) and \(\mathbf1\) give a global linear functional
\(
    p(v)\cdot x
\)
with \(p(v)\in\Delta(S)\) by statewise monotonicity. Continuity of the
family follows directly from \citet{Payro2025}'s continuity in the level argument: for
each state $s$, writing $\mathbf e_s$ for the singleton indicator gives
$ p_s(v)=\Phi(\mathbf e_s,v)$, so every coordinate $p_s(v)$ is continuous
on $(0,1)$. The additive capacity
\(
    \widetilde\nu_v(A):=\sum_{s\in A}p_s(v)
\)
therefore generates the same level-\(v\) threshold representation. Uniqueness
in Proposition \ref{prop:revealed-capacity} implies
\(\nu_v=\widetilde\nu_v\).

Conversely, suppose every \(\nu_v\) is additive, so
\(
    C_{\nu_v}(x)=p(v)\cdot x
\)
is affine in \(x\) on the whole utility domain. Let $a:=V(x)$,
$b:=V(y)$, and $z_\lambda:=\lambda x+(1-\lambda)y$ for
$\lambda\in(0,1)$. If $1>a>b>0$, then the threshold representation gives
\(
    p(a)\cdot x=a
    \;\text{and}\;
    p(a)\cdot y<a,
\)
and therefore
\(
    p(a)\cdot z_\lambda<a.
\)
Likewise,
\(
    p(b)\cdot x>b
    \; \text{and}\;
    p(b)\cdot y=b,
\)
so $p(b)\cdot z_\lambda>b$. The single-crossing orientation then yields
\(
    b<V(z_\lambda)<a.
\)
If $a=1>b>0$, then $x=\mathbf1$; the level-$b$ inequality gives
$V(z_\lambda)>b$, while $z_\lambda\neq\mathbf1$ implies
$V(z_\lambda)<1=a$. The case $1>a>b=0$ is symmetric, and if
$a=1$, $b=0$, then $z_\lambda=\lambda\mathbf1$ and
$V(z_\lambda)=\lambda$. Finally, if $a=b=v\in(0,1)$, global affinity gives
\(
    p(v)\cdot z_\lambda
    =\lambda v+(1-\lambda)v
    =v,
\)
so uniqueness of the implicit root gives $V(z_\lambda)=v$; the endpoint
equality cases are immediate. Hence mixture-betweenness holds for arbitrary
pairs in $X$.
\end{proof}

\begin{proof}[Proof of Proposition \ref{prop:fixed-capacity}]
Suppose first that preferences satisfy comonotonic independence. We first show
that the certainty-equivalent map is translation invariant. Fix $x\in X$ and
$t$ such that $x+t\mathbf1\in X$. The case $t=0$ is immediate.

Suppose $t\in(0,1)$. Feasibility gives
$\max_s x_s\leq1-t$. Put $\alpha:=1-t$ and define
\(f:=\frac{x}{\alpha}, \; g:=\frac{V(x)}{\alpha}\mathbf1\).
Statewise monotonicity gives
$V(x)\leq\max_s x_s\leq\alpha$, so $f,g\in X$. Moreover,
\(x=\alpha f+(1-\alpha)\mathbf0, \; V(x)\mathbf1=\alpha g+(1-\alpha)\mathbf0\).
Since $x\sim V(x)\mathbf1$, comonotonic independence applied in both
directions to the pairwise comonotonic acts $f,g,\mathbf0$ gives $f\sim g$.
Applying comonotonic independence again, now with common component
$\mathbf1$, yields
\(x+t\mathbf1 =\alpha f+(1-\alpha)\mathbf1 \sim \alpha g+(1-\alpha)\mathbf1 =[V(x)+t]\mathbf1\).
Uniqueness of the certainty equivalent therefore gives
\(
    V(x+t\mathbf1)=V(x)+t.
\)

If $t=-s$ with $s\in(0,1)$, feasibility gives
$\min_s x_s\geq s$, and hence $V(x)\geq s$. Put
$\alpha:=1-s$ and define
\(f:=\frac{x-s\mathbf1}{\alpha}, \; g:=\frac{V(x)-s}{\alpha}\mathbf1\).
Then $f,g\in X$ and
\(x=\alpha f+(1-\alpha)\mathbf1, \; V(x)\mathbf1=\alpha g+(1-\alpha)\mathbf1\).
From $x\sim V(x)\mathbf1$, comonotonic independence with common component
$\mathbf1$ gives $f\sim g$. Replacing the common component by $\mathbf0$
therefore gives
\(x-s\mathbf1 =\alpha f+(1-\alpha)\mathbf0 \sim \alpha g+(1-\alpha)\mathbf0 =[V(x)-s]\mathbf1\).
Thus $V(x-s\mathbf1)=V(x)-s$. The endpoint translations $|t|=1$ can occur
only between $\mathbf0$ and $\mathbf1$ and are immediate. Hence $V$ is
certainty translation invariant on its feasible domain.

We now show directly that the capacity family is constant. Fix interior
$v,w$, an event $A\subseteq S$, and put $c:=w-v$ and
$r:=\nu_v(A)$. Choose $\varepsilon>0$ sufficiently small that the binary act
below and its translate by $c\mathbf1$ both belong to $X$, and set
\(a:=\varepsilon(1-r), \; b:=\varepsilon r, \; x:=(v+a)\mathbf1_A+(v-b)\mathbf1_{A^c}\).
Here $a,b\geq0$ and $a+b=\varepsilon>0$. The binary-act formula gives
\(C_{\nu_v}(x) =v-b+(a+b)r =v\).
Since $C_{\nu_v}(x)=v$, the equality threshold relation gives
$x\sim v\mathbf1$, so $V(x)=v$. Translation invariance then gives
\(x+c\mathbf1\sim w\mathbf1\).
The translated act has the same increments $a,b$ around $w$. Applying the
same level-$w$ equality relation and the binary Choquet formula gives
\(-b+(a+b)\nu_w(A)=0\).
Therefore
\(\nu_w(A)=\frac{b}{a+b}=r=\nu_v(A)\).
This argument also covers $r\in\{0,1\}$, in which case one of $a,b$ is zero.
Since $A,v,w$ were arbitrary, $\nu_v=\nu_w$ for all interior levels. Writing
the common capacity as $\nu$, Theorem \ref{thm:implicit-choquet} gives
\(
    V(x)=C_\nu(x)
\)
for every $x\in X$.

For any nontrivial event $A$, the indicator $\mathbf1_A$ is a nonconstant
boundary act. Axiom \ref{ax:boundary-interiority} and
$V(\mathbf1_A)=C_\nu(\mathbf1_A)=\nu(A)$ therefore imply
\(
    0<\nu(A)<1.
\)

Conversely, suppose $\nu_v=\nu$ for every interior $v$, with
$0<\nu(A)<1$ for every nontrivial event. Theorem
\ref{thm:implicit-choquet} gives $V(x)=C_\nu(x)$ on $X$. If $x,y,z$ are
pairwise comonotonic, comonotonic affinity of the Choquet integral gives
\(C_\nu\!\left(\alpha x+(1-\alpha)z\right) =\alpha C_\nu(x)+(1-\alpha)C_\nu(z)\)
and
\(C_\nu\!\left(\alpha y+(1-\alpha)z\right) =\alpha C_\nu(y)+(1-\alpha)C_\nu(z)\).
Subtracting the two affine identities shows that the sign of
\(C_\nu(x)-C_\nu(y)\) is preserved when both acts are mixed with the common
component $z$, which is comonotonic independence. The strict event bounds are exactly the
maintained boundary nondegeneracy in the fixed-capacity subclass.
\end{proof}

\begin{proof}[Proof of Corollary \ref{cor:seu-intersection}]
Proposition~\ref{prop:global-betweenness} implies that global
mixture-betweenness is equivalent to additivity of every local capacity:
\(
    \nu_v(A)=\sum_{s\in A}p_s(v)
\)
for some probability vector $p(v)$. Proposition~\ref{prop:fixed-capacity}
implies that comonotonic independence is equivalent to independence of the
capacity from the certainty-equivalent level. Imposing both therefore gives a single
probability vector $p\in\Delta(S)$ such that
\(
    \nu_v(A)=\sum_{s\in A}p_s
    \;\text{for every }A\subseteq S\text{ and }v\in(0,1).
\)
Proposition \ref{prop:fixed-capacity} implies that every singleton has
strictly positive weight, so $p$ has full support. The Choquet integral of an
additive capacity is the ordinary expectation, so
Theorem~\ref{thm:implicit-choquet} reduces to
\(
    V(x)=p\cdot x
\), which is full-support subjective expected utility on the maintained affine
consequence scale. Conversely, any full-support subjective expected utility is
globally affine in mixtures, satisfies both global mixture-betweenness and
comonotonic independence, and satisfies Axiom \ref{ax:boundary-interiority}.
\end{proof}

\begin{proof}[Proof of Proposition \ref{prop:local-ua}]
Suppose first that $C_{\nu_v}$ is concave. If $x,y\sim v\mathbf1$, then
Theorem \ref{thm:implicit-choquet} gives
\(
    C_{\nu_v}(x)=C_{\nu_v}(y)=v.
\)
Hence
\(
    C_{\nu_v}(\lambda x+(1-\lambda)y)
    \geq
    \lambda v+(1-\lambda)v
    =v,
\)
and the threshold representation implies \( \lambda x+(1-\lambda)y
    \succeq
    v\mathbf1\).

Conversely, suppose local uncertainty aversion holds. Take arbitrary
$x,y\in X$ and put
\(
    a:=C_{\nu_v}(x)
    \;\text{and}\;
    b:=C_{\nu_v}(y).
\)
Because $v$ is interior, for sufficiently small $\varepsilon>0$ the profiles
\(
    \widetilde x
    :=v\mathbf1+\varepsilon(x-a\mathbf1),
    \;\text{and}\;
    \widetilde y
    :=v\mathbf1+\varepsilon(y-b\mathbf1)
\)
belong to $X$. Using the canonical extension of the Choquet integral to
$\mathbb R^S$ stated after \eqref{eq:choquet-integral}---only the intermediate
centered vectors $x-a\mathbf1$ and $y-b\mathbf1$ need lie outside $X$---its
translation invariance and positive homogeneity give
\(
    C_{\nu_v}(\widetilde x)
    =C_{\nu_v}(\widetilde y)
    =v.
\)
Thus both profiles are indifferent to $v\mathbf1$. Local uncertainty
aversion implies
\(
    C_{\nu_v}
    \bigl(\lambda\widetilde x+(1-\lambda)\widetilde y\bigr)
    \geq v.
\)
Using the same extended-domain identities once more and dividing by
$\varepsilon$ yields
\(
    C_{\nu_v}(\lambda x+(1-\lambda)y)
    \geq
    \lambda C_{\nu_v}(x)+(1-\lambda)C_{\nu_v}(y).
\)
Hence $C_{\nu_v}$ is concave. For finite capacities, concavity of the
Choquet integral is equivalent to supermodularity of the capacity, giving
(i)--(iii).

If local uncertainty neutrality holds, the same scaling argument gives both
concavity and convexity of $C_{\nu_v}$, hence affinity. A finite Choquet
integral is affine on the whole domain if and only if its capacity is additive.
The converse is immediate.
\end{proof}

\begin{proof}[Proof of Corollary \ref{cor:implicit-multiple-priors}]
By Proposition \ref{prop:local-ua}, each $\nu_v$ is convex. The standard
core representation of a convex capacity gives
\(
    C_{\nu_v}(x)
    =
    \min_{p\in\operatorname{core}(\nu_v)}p\cdot x.
\)
Substituting this equality into Theorem \ref{thm:implicit-choquet} gives the
result.
\end{proof}

\begin{proof}[Proof of Proposition \ref{prop:core-expansion}]
For a convex capacity,
\(
    \nu_z(A)=\min_{p\in\mathcal P_z}p(A)
    \;
    \text{for every }A.
\)
Moreover, by definition,
\(
    \mathcal P_z
    =
    \{p\in\Delta(S):p(A)\geq\nu_z(A)\text{ for all }A\}.
\)
Hence (i) relaxes every lower-probability constraint defining
$\mathcal P_v$, which is equivalent to $\mathcal P_v\subseteq\mathcal P_w$.
The core representation then gives
\(
    C_{\nu_w}(x)
    =\min_{p\in\mathcal P_w}p\cdot x
    \leq
    \min_{p\in\mathcal P_v}p\cdot x
    =C_{\nu_v}(x),
\)
so (ii) implies (iii). Finally, applying (iii) to indicator acts
$\mathbf1_A$ gives
$\nu_w(A)\leq\nu_v(A)$, proving (iii) implies (i).
\end{proof}

\begin{proof}[Proof of Proposition \ref{prop:translation-fixed}]
If (ii) holds, writing the common capacity as $\nu$ gives (iii), and
translation invariance of the Choquet integral then gives (iii)$\Rightarrow$(i).
It remains to show (i)$\Rightarrow$(ii); together these implications establish
the equivalence.

Fix interior $v,w$ and an event $A$, and put $c:=w-v$. Let
$r:=\nu_v(A)$. Choose $\varepsilon>0$ small enough that the following binary
act and its translate remain in $X$, and set
\(a:=\varepsilon(1-r), \; b:=\varepsilon r, \; x:=(v+a)\mathbf1_A+(v-b)\mathbf1_{A^c}\).
Here $a,b\geq0$ and $a+b=\varepsilon>0$; when $r\in\{0,1\}$ one increment is
zero, a case explicitly covered by Proposition \ref{prop:revealed-capacity}.
The binary-act formula therefore gives
\(
    C_{\nu_v}(x)
    =v-b+(a+b)r
    =v,
\)
so $V(x)=v$. Translation invariance implies
$V(x+c\mathbf1)=w$. The translated act has exactly the same increments
$a,b$ around $w$. Evaluating its level-$w$ indifference gives
\(
    -b+(a+b)\nu_w(A)=0.
\)
Thus $\nu_w(A)=b/(a+b)=r=\nu_v(A)$. Since $A,v,w$ were arbitrary,
$\nu_v$ is constant across certainty-equivalent levels.
\end{proof}

\begin{proof}[Proof of Corollary \ref{cor:wci-fixed}]
We first show that Weak Certainty Independence implies certainty translation
invariance. Fix $x\in X$ and $t$ such that $x+t\mathbf1\in X$. The case
$t=0$ is immediate.

Suppose first that $t\in(0,1)$. Feasibility gives
$\max_s x_s\leq1-t$, and therefore $V(x)\leq1-t$. Put
$\alpha:=1-t$ and define
\(f:=\frac{x}{\alpha}, \; g:=\frac{V(x)}{\alpha}\mathbf1\).
Then $f,g\in X$ and
\(
    x=\alpha f+(1-\alpha)\mathbf0
    \; \text{and}\;
    V(x)\mathbf1=\alpha g+(1-\alpha)\mathbf0.
\)
Since $x\sim V(x)\mathbf1$, applying Weak Certainty Independence in both
directions and replacing the common constant $\mathbf0$ by $\mathbf1$
gives
\(
    x+t\mathbf1
    =\alpha f+(1-\alpha)\mathbf1
    \sim
    \alpha g+(1-\alpha)\mathbf1
    =[V(x)+t]\mathbf1.
\)
Uniqueness of the certainty equivalent therefore yields
$V(x+t\mathbf1)=V(x)+t$.

If $t=-s$ with $s\in(0,1)$, feasibility gives
$\min_s x_s\geq s$ and hence $V(x)\geq s$. Put
$\alpha:=1-s$ and define
\(f:=\frac{x-s\mathbf1}{\alpha}, \; g:=\frac{V(x)-s}{\alpha}\mathbf1\).
Then $f,g\in X$ and
\(
    x=\alpha f+(1-\alpha)\mathbf1
    \; \text{and}\;
    V(x)\mathbf1=\alpha g+(1-\alpha)\mathbf1.
\)
Applying Weak Certainty Independence in both directions and replacing the
common constant $\mathbf1$ by $\mathbf0$ gives
\(
    x-s\mathbf1\sim[V(x)-s]\mathbf1,
\)
so again $V(x-s\mathbf1)=V(x)-s$. The endpoint translations $|t|=1$
can occur only between $\mathbf0$ and $\mathbf1$ and are immediate. Thus
$V$ satisfies certainty translation invariance.

Proposition~\ref{prop:translation-fixed} now implies that
$\nu_v=\nu$ for all $v\in(0,1)$ and $V(x)=C_\nu(x)$ on $X$.
Finally, Axiom~\ref{ax:boundary-interiority} applied to indicator acts gives
$0<\nu(A)<1$ for every nontrivial event $A$.
\end{proof}


\begin{thebibliography}{99}

\bibitem[Anscombe and Aumann(1963)]{AnscombeAumann1963}
Anscombe, F. J. and R. J. Aumann (1963),
``A Definition of Subjective Probability,''
\emph{The Annals of Mathematical Statistics}, 34(1), 199--205.

\bibitem[Baillon and Bleichrodt(2015)]{BaillonBleichrodt2015}
Baillon, A. and H. Bleichrodt (2015),
``Testing Ambiguity Models through the Measurement of Probabilities for Gains
and Losses,''
\emph{American Economic Journal: Microeconomics}, 7(2), 77--100.

\bibitem[Baillon et al.(2025)]{BaillonBleichrodtLiWakker2025}
Baillon, A., H. Bleichrodt, C. Li, and P. P. Wakker (2025),
``Source Theory: A Tractable and Positive Ambiguity Theory,''
\emph{Management Science}, 71(10), 8767--8782.

\bibitem[Baillon and Placido(2019)]{BaillonPlacido2019}
Baillon, A. and L. Placido (2019),
``Testing Constant Absolute and Relative Ambiguity Aversion,''
\emph{Journal of Economic Theory}, 181, 309--332.

\bibitem[Bouchouicha et al.(2017)]{BouchouichaEtAl2017}
Bouchouicha, R., P. Martinsson, H. Medhin, and F. M. Vieider (2017),
``Stake Effects on Ambiguity Attitudes for Gains and Losses,''
\emph{Theory and Decision}, 83(1), 19--35.

\bibitem[Cerreia-Vioglio et al.(2011)]{CerreiaVioglioEtAl2011}
Cerreia-Vioglio, S., F. Maccheroni, M. Marinacci, and L. Montrucchio (2011),
``Uncertainty Averse Preferences,''
\emph{Journal of Economic Theory}, 146(4), 1275--1330.

\bibitem[Cerreia-Vioglio, Maccheroni, and Marinacci(2022)]{CerreiaVioglioMaccheroniMarinacci2022}
Cerreia-Vioglio, S., F. Maccheroni, and M. Marinacci (2022),
``Ambiguity Aversion and Wealth Effects,''
\emph{Journal of Economic Theory}, 199, 104898.

\bibitem[Chateauneuf and Faro(2009)]{ChateauneufFaro2009}
Chateauneuf, A. and J. H. Faro (2009),
``Ambiguity through Confidence Functions,''
\emph{Journal of Mathematical Economics}, 45(9--10), 535--558.

\bibitem[Chew and Epstein(1989)]{ChewEpstein1989}
Chew, S. H. and L. G. Epstein (1989),
``A Unifying Approach to Axiomatic Non-Expected Utility Theories,''
\emph{Journal of Economic Theory}, 49(2), 207--240.

\bibitem[Chew, Epstein, and Wakker(1993)]{ChewEpsteinWakker1993}
Chew, S. H., L. G. Epstein, and P. P. Wakker (1993),
``A Unifying Approach to Axiomatic Non-Expected Utility Theories: Correction
and Comment,''
\emph{Journal of Economic Theory}, 59(1), 183--188.

\bibitem[Chew and Karni(1994)]{ChewKarni1994}
Chew, S. H. and E. Karni (1994),
``Choquet Expected Utility with a Finite State Space: Commutativity and
Act-Independence,''
\emph{Journal of Economic Theory}, 62(2), 469--479.

\bibitem[Chew and Wakker(1996)]{ChewWakker1996}
Chew, S. H. and P. P. Wakker (1996),
``The Comonotonic Sure-Thing Principle,''
\emph{Journal of Risk and Uncertainty}, 12(1), 5--27.

\bibitem[Dekel(1986)]{Dekel1986}
Dekel, E. (1986),
``An Axiomatic Characterization of Preferences under Uncertainty: Weakening
the Independence Axiom,''
\emph{Journal of Economic Theory}, 40(2), 304--318.

\bibitem[Denti and Pomatto(2022)]{DentiPomatto2022}
Denti, T. and L. Pomatto (2022),
``Model and Predictive Uncertainty: A Foundation for Smooth Ambiguity Preferences,''
\emph{Econometrica}, 90(2), 551--584.

\bibitem[Ghirardato, Maccheroni, and Marinacci(2004)]{GhirardatoMaccheroniMarinacci2004}
Ghirardato, P., F. Maccheroni, and M. Marinacci (2004),
``Differentiating Ambiguity and Ambiguity Attitude,''
\emph{Journal of Economic Theory}, 118(2), 133--173.

\bibitem[Ghirardato and Marinacci(2002)]{GhirardatoMarinacci2002}
Ghirardato, P. and M. Marinacci (2002),
``Ambiguity Made Precise: A Comparative Foundation,''
\emph{Journal of Economic Theory}, 102(2), 251--289.

\bibitem[Gilboa and Schmeidler(1989)]{GilboaSchmeidler1989}
Gilboa, I. and D. Schmeidler (1989),
``Maxmin Expected Utility with Non-Unique Prior,''
\emph{Journal of Mathematical Economics}, 18(2), 141--153.

\bibitem[Grant, Kajii, and Polak(2000)]{GrantKajiiPolak2000}
Grant, S., A. Kajii, and B. Polak (2000),
``Decomposable Choice under Uncertainty,''
\emph{Journal of Economic Theory}, 92(2), 169--197.

\bibitem[Grant and Polak(2013)]{GrantPolak2013}
Grant, S. and B. Polak (2013),
``Mean-Dispersion Preferences and Constant Absolute Uncertainty Aversion,''
\emph{Journal of Economic Theory}, 148(4), 1361--1398.

\bibitem[Grant, Roorda, and Yang(2025)]{GrantRoordaYang2025}
Grant, S., B. Roorda, and J. Yang (2025),
``Expected Balanced Uncertain Utility,''
\emph{Theoretical Economics}, 20(1), 1--25.

\bibitem[Hazen(1987)]{Hazen1987}
Hazen, G. B. (1987),
``Subjectively Weighted Linear Utility,''
\emph{Theory and Decision}, 23(3), 261--282.

\bibitem[Klibanoff, Marinacci, and Mukerji(2005)]{KlibanoffMarinacciMukerji2005}
Klibanoff, P., M. Marinacci, and S. Mukerji (2005),
``A Smooth Model of Decision Making under Ambiguity,''
\emph{Econometrica}, 73(6), 1849--1892.

\bibitem[K\"onig-Kersting, Kops, and Trautmann(2023)]{KonigKerstingKopsTrautmann2023}
K\"onig-Kersting, C., C. Kops, and S. T. Trautmann (2023),
``A Test of (Weak) Certainty Independence,''
\emph{Journal of Economic Theory}, 209, 105623.

\bibitem[Maccheroni, Marinacci, and Rustichini(2006)]{MaccheroniMarinacciRustichini2006}
Maccheroni, F., M. Marinacci, and A. Rustichini (2006),
``Ambiguity Aversion, Robustness, and the Variational Representation of Preferences,''
\emph{Econometrica}, 74(6), 1447--1498.

\bibitem[Marinacci, Principi, and Stanca(2026)]{MarinacciPrincipiStanca2026}
Marinacci, M., G. Principi, and L. Stanca (2026),
``Recursive Preferences and Ambiguity Attitudes,''
\emph{Journal of Economic Theory}, 236, 106224.

\bibitem[Payr\'o(2025)]{Payro2025}
Payr\'o, F. (2025),
``Mixture-Betweenness: Uncertainty and Commitment,''
\emph{Journal of Economic Theory}, 230, 106097.

\bibitem[Payr\'o, Takeoka, and Xia(2026)]{PayroTakeokaXia2026}
Payr\'o, F., N. Takeoka, and J. Xia (2026),
``A Cognitive Theory of Ambiguity Attitudes,''
BSE Working Paper 1587, Barcelona School of Economics, July 2026 version.

\bibitem[Roorda and Joosten(2026)]{RoordaJoosten2026}
Roorda, B. and R. Joosten (2026),
``Balancing Acts,''
\emph{Journal of Mathematical Economics}, 126, 103273.

\bibitem[Schmeidler(1989)]{Schmeidler1989}
Schmeidler, D. (1989),
``Subjective Probability and Expected Utility without Additivity,''
\emph{Econometrica}, 57(3), 571--587.

\bibitem[Siniscalchi(2009)]{Siniscalchi2009}
Siniscalchi, M. (2009),
``Vector Expected Utility and Attitudes Toward Variation,''
\emph{Econometrica}, 77(3), 801--855.

\bibitem[Trautmann and Wakker(2018)]{TrautmannWakker2018}
Trautmann, S. T. and P. P. Wakker (2018),
``Making the Anscombe--Aumann Approach to Ambiguity Suitable for Descriptive
Applications,''
\emph{Journal of Risk and Uncertainty}, 56(1), 83--116.

\end{thebibliography}
\end{document}